\documentclass[11pt]{article}
\usepackage{graphics}
\usepackage{latexsym}
\usepackage{amsfonts}
\usepackage{epstopdf}
\usepackage{appendix}
\usepackage{caption}
\usepackage{subcaption}
\usepackage{graphicx}
\usepackage{amsmath}
\usepackage{amsthm}
\usepackage{amssymb}
\newtheorem{definition}{Definition}
\newtheorem{theorem}{Theorem}

\newtheorem{lemma}{Lemma}

\newtheorem{proposition}{Proposition}
\newtheorem{remark}{Remark}

\begin{document}
\begin{titlepage}
\Large
\begin{center}
				{\bf \Huge
				A Passivity-Based Design for Stability and Robustness in Event-Triggered Networked Control Systems with Communication Delays, Signal Quantizations and Packet Dropouts
				}
\end{center}
\begin{center}
				Technical Report of the ISIS Group\\
				at the University of Notre Dame\\
				ISIS-2016-003\\ 
				May 2016 
\end{center}
\vspace{1.5in}		
\begin{center}
				Arash Rahnama$^1$, Meng Xia$^1$ and Panos J. Antsaklis$^1$\\
				$^1$Department of Electrical Engineering\\
				University of Notre Dame\\
				Notre Dame, IN 46556\\
\end{center}
\vspace{.8in}
\begin{center}
				{\bf \Large Interdisciplinary Studies in Intelligent Systems}
\end{center}
\vspace{0.5in}
\begin{center}
				{\textbf {Acknowledgements} The support of the National Science Foundation under the CPS Grant No. CNS-1035655 and CNS-1446288 is gratefully acknowledged.
				}
\end{center}
\end{titlepage}
\clearpage
		
\begin{abstract}
		In this report, we introduce a comprehensive design framework for Event-Triggered Networked Control Systems based on the passivity-based concept of Input Feed-Forward Output Feedback Passive (IF-OFP) systems. Our approach is comprehensive in the sense that we show finite-gain $L_2$-stability and robustness for the networked control system by considering the effects of time-varying or constant network induced delays, signal quantizations, and data losses in communication links from the plant to controller and the controller to plant. Our design is based on the need for a more efficient utilization of band-limited shared communication networks which is a necessity for the design of Large-Scale Cyber-Physical systems. To achieve this, we introduce simple triggering conditions that do not require the exact knowledge of the sub-systems and are located on both sides of the communication network: the plant's output and the controller's output. This specifically leads to a great decrease in the communication rate between the controller and plant. Additionally, we show lower-bounds on inter-event time intervals for the triggering conditions and show the design's robustness against external noise and disturbance. We illustrate the relationship amongst stability, robustness and passivity levels for the plant and controller. We analyze our design's robustness against packet dropouts and loss of communication. Our results are design-oriented in the sense that based on our proposed framework, the designer can easily observe the trade-offs amongst different components of the networked control system, time-varying delays, effects of signal quantizations and triggering conditions, stability, robustness and performance of networked control system and make design decisions accordingly.
\end{abstract}
\clearpage

\section{Introduction}
Networked Control Systems (NCS) provide the possibility of controlling different sub-systems through unique remote controllers that may be situated at different locations away from each other. Under this framework, a shared multi-purpose communication network provides the platform on which, the controllers and plants exchange information to further the overall goal of large-scale compositional control systems such as cyber-physical systems (CPS). NCSs are spatially distributed and the communication amongst sensors, actuators and controllers happen through a shared communication network (figure \ref{fig:NCS}). The idea that one can design multi-purpose controllers that are located remotely and are able to satisfy complicated application goals and needs by utilizing different sub-systems and plants is extremely desirable and even necessary for many of today's technological advances and necessities. NSCs provide flexibility and low cost of design, efficiency and ease of maintenance. As a result, they are amongst one of the top topics that recently have been enjoying the research spotlight in fields such as control and communication theory   \cite{hespanha2007survey,baillieul2007control}. Additionally, NCS designs and their benefits are utilized in a large range of areas such as  unmanned aerial vehicles \cite{seiler2005h}, remote surgery \cite{meng2004remote}, mobile sensor networks \cite{ogren2004cooperative}, and haptics collaboration over the internet\cite{ogren2004cooperative,hikichi2002evaluation}. A comprehensive survey on recent developments in this field can be found in \cite{hespanha2007survey, baillieul2007control}.

Applying communication links between the controller and plant - as opposed to traditional dedicated connections - to achieve control goals brings about several challenges such as information loss, band-limited channels and  limited communication capabilities for shared networks, delays and signal quantizations. Recent technological advances which has led to the production of cheap, small and fast microprocessors with powerful control capabilities, however, has made it possible for engineers to try and solve these challenges. Delays in NCS' communication links can be highly variable due to variant network access times and transmission delays. Some results in the literature dealing with the problem of delays propose upper bounds for allowable delays called the maximum allowable transfer interval (MATI) \cite{walsh2002stability}. In \cite{gao2008new}, authors try to model the delays in a networked control system to produce a  Lyapunov-based method with an acceptable $H^{\infty}$ performance index. In \cite{lam2007stability}, authors come up with LMI-based solutions for additive time-varying delays in networked control systems - some other approaches for dealing with time delays in NCS are given in \cite{branicky2000stability,lin2003robust}. The problem of lost information, or packet dropouts and signal quantization have also been addressed in the literature. In \cite{tsumura2009tradeoffs}, authors discuss the trade-off between quantization and packet dropouts - a similar line of work is presented in \cite{qu2015event}. The author in \cite{fu2008quantization} provides a thorough analysis of negative effects of quantization on stability and estimation in networked control systems. An $H^{\infty}$ output feedback control method is given in \cite{jiang2013h} for nonlinear networked control systems with delays and packet dropouts. An LMI-based approach for dealing with packet dropouts and delays in networked control systems was proposed in \cite{yu2004lmi}. In \cite{xiong2007stabilization}, authors address the stabilization of networked control systems in the presence of arbitrary packet dropouts or packet dropouts that follow Markovian patterns, and propose a Lyapunov-based solution.

Classical networked control designs rely on the traditional periodic task of control \cite{ragazzini1958sampled, aastrom2013computer}. However, the results under this approach are quite conservative and require an abundance of recourses in terms of CPU usage, and communication rate. Consequently,  event-triggered control has been proposed as a more efficient alternative. In contrast to the classic periodical sensing and actuation control, under an event-based framework, information between the plant and controller is only exchanged when it is necessary. This usually happens when a certain controlled value in the system deviates from its desired value for larger than a certain threshold. In other words, the information is exchanged only when something "significant" happens, and control is not executed unless it is required. Under this framework, one can obtain most control objectives by an open loop controller, while uncertainty is inevitable in real systems, a close-loop event-based framework can robustly deal with these uncertainties. As a result, event-based control over networks has regained research interest since it creates a better balance between control performance, communication, and computational load  compared to the time-based counterpart \cite{aastrom1999comparison}. Similar works in the literature under different names but based on the same core concept and only  with slight differences are event-based sampling \cite{astrom2008event}, event-driven sampling \cite{heemels2008analysis}, state-triggering sampling \cite{tabuada2007event} and self-triggering sampling \cite{wang2010self, wang2011event}. A comprehensive survey on event-triggered control is given in \cite{lemmon2010event}. 

Most works in the field of event-triggered networked control can be divided into two categories: 

1) Stability in terms of input-to-state or ISS-stability \cite{sontag1989smooth,tabuada2007event, anta2010sample}, where the connection between stability and states of the system is explored\textemdash This approach relies on the full knowledge of states of each sub-system which can be impossible in some cases. 

2) Stability in terms of finite-gain input-output stability \cite{yu2013event,kofman2006level, zhu2014passivity, yu2011event}.

A lot of works in event-triggered networked control literature assume perfect communication links with no delays, and fail to take into consideration the effects of signal quantizations. Additionally, the robustness issues such as packet dropouts, disturbance and a full analysis of inter-event time intervals \textemdash insurance against continuous exchange of data between sub-units or what is known as "Zeno" behavior\textemdash are not fully explored and understood in the literature for event-triggered networked control systems.  Time-delays are considered only when there is an upper-bound assumptions on them which can be seen as unrealistic in practical applications. In this report, we take advantage of the theory of dissipativity, \emph{QSR-dissipativity} and passivity to design robust event-triggered networked control systems. Passivity and dissipativity encompass the energy consumption characterizations of a dynamical system. Passivity is preserved under parallel and feedback interconnections \cite{bao2007process}. Passivity also implies stability under mild assumptions \cite{bao2007process,khalil2002nonlinear} making them a great alternative for designing compositional large-scale networked control systems.

Motivated by the previous works of our colleagues in \cite{yu2011event,yu2013event}, in this report, we will propose an event-triggered networked control system design that can ensure $L_2$-stability in the presence of signal quantizations, packet dropouts, and time-varying or constant delays. Similar to our previous work in \cite{rahnama2016qsr}, we consider a large class of systems known as Input Feed-Forward Output Feedback Passive (IF-OFP) systems, and derive triggering conditions based on a passivity-based approach which allows us to include a big class of systems with negative or positive known passivity indices. Our design includes event-triggering conditions that are located on both the plant's output side, and controller's output side. This leads to a great decrease in the required rate of communication amongst sub-systems, and lessens the communication load on the shared band-limited network. 

In our previous work in \cite{rahnama2016qsr}, we analyzed the \emph{QSR-dissipativity} and passivity for event-triggered networked control systems: we showed \emph{QSR-dissipativity} and stability conditions for different triggering frameworks based on theorem 1, and calculated passivity indices for the networked control interconnections with different triggering conditions. While in our work in \cite{rahnama2016qsr} we included the negative effects of disturbance and noise, we did not take into account the effects of signal quantizations, and time-delays in networked control frameworks. In our current work, we intend to propose a more comprehensive design that considers the effects of signal quantizations, constant or time-varying communication delays, and packet dropouts.

Under our proposed design, finite-gain $L_2$-stability is guaranteed and a transformation matrix $M$ is applied on the plant's side of the communication links to deal with the presence of signal quantizations and constant or time-varying delays with bounded rates of change. The delays existing in both communication links: from the plant to controller and from the controller to plant, are considered. In addition, these delays can be larger than inter-event time intervals. Our intention is to design the $M$-Transformation matrix such that the negative effects of signal quantizations, packet dropouts and time delays on both sides are overcome, and as a result the entire system is both finite-gain $L_2$-stable and robust. Our work is influenced by the study done in \cite{yu2013event}, however we seek to implement a design-based approach with easily understood and simple triggering conditions that are located on both sides of the networked control system. There are two benefits to this approach: firstly, the designer will be able to see a clear trade-off between the communication rate and performance and make design decisions accordingly; secondly, the triggering conditions are located on both sides of the network, this tremendously decreases the communication load on the shared band-limited network, specifically we show that we can greatly decrease the rate of communication from the controller to plant. This can help with the efficiency of the ever-increasing number of sub-units in large-scale CPS that share the same communication network.

Additionally, our work can be applied to linear and nonlinear systems and does not rely on the exact knowledge of sub-systems; only an access to inputs and outputs of sub-systems is enough. Moreover, we address robustness issues with respect to imperfect communication links and network uncertainties. We show a lower-bound (lack of Zeno-effect) for inter-event time intervals based on conic properties of IFOF passive systems. Our approach requires no additional assumptions on the linear or nonlinear system's behavior. These results also clarify the robustness of each triggering condition against external noise and disturbance. We also analyze the performance of our design under data losses and packet dropouts. We show a trade-off between the design parameters and the number of allowable consecutive packet dropouts for maintaining finite-gain $L_2$-stability.   

This report is organized as the following: section \ref{sec:qsrpassivty} gives the preliminarily mathematical definitions on dissipativity and passivity. The problem statement being addressed is clearly stated in section \ref{sec:prob}. Section \ref{sec:structure} illustrates the general networked control framework \textemdash including the event-triggering conditions, signal quantizations, and time-delays\textemdash  that our work is based on.  Section \ref{sec:main} includes our results on passivity, event-triggering conditions, and $L_2$-stabilty and also the proposed transformation matrix design. In section \ref{sec:robustness} a thorough analysis of inter-event time intervals based on conic properties of IFOF systems, and robustness issues with respect to packet dropouts, and external disturbances is given. An illustrative example is given in section \ref{sec:simulation}. Section \ref{sec:conclusion} gives an overview of what was achieved in the report and concludes the report.  
\clearpage

\section{Mathematical Preliminaries}

\begin{figure}[!t]
	\centering
	\includegraphics[scale = 0.7]{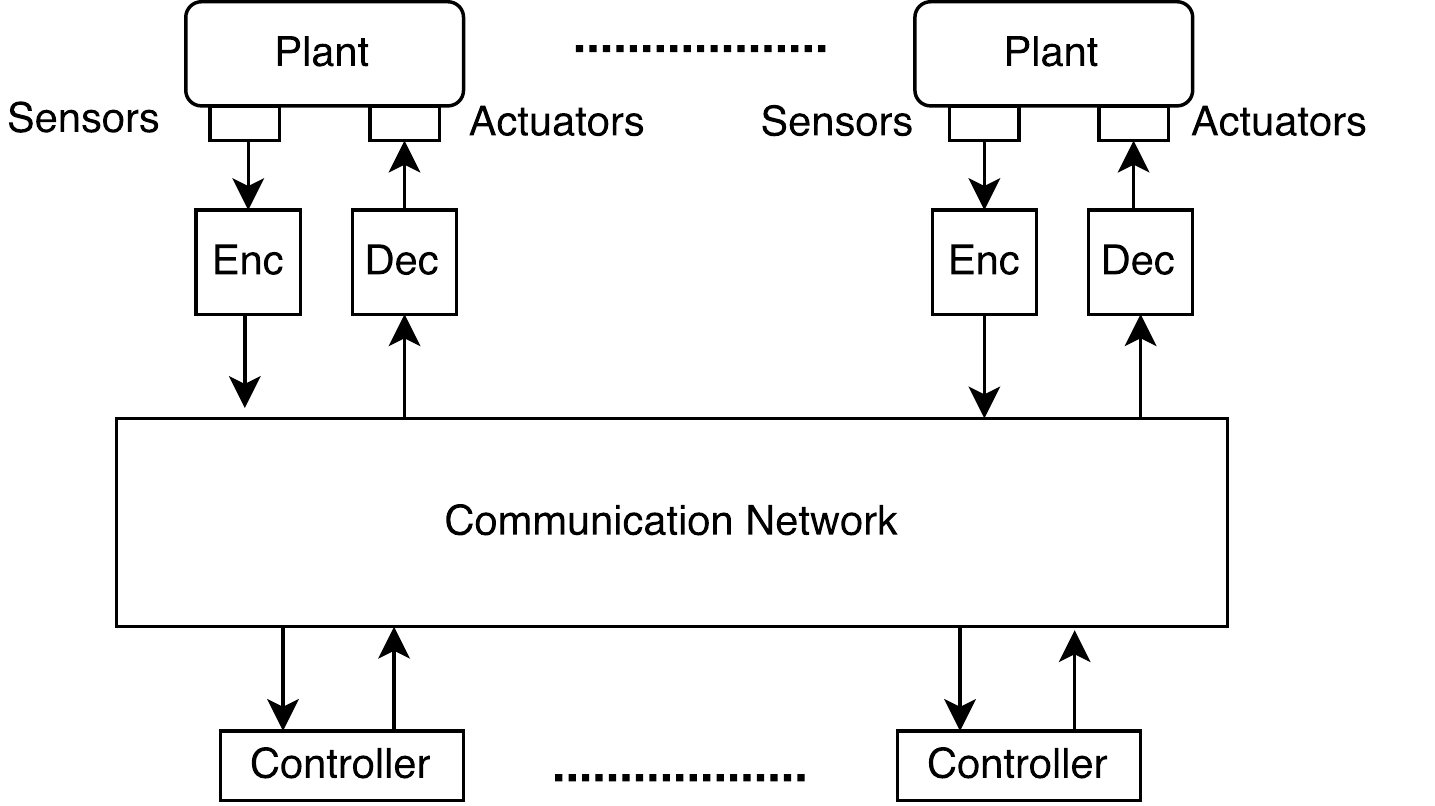}
	\caption{A Networked Control System Framework.}
	\label{fig:NCS}
\end{figure}   

\label{sec:qsrpassivty}
Consider the following linear or nonlinear dynamical system $G$, 
\begin{align}
	\label{dynsys}
	G:
	\begin{cases}
	\dot{x}(t)= f(x(t),u(t)) &  \\
	y(t)= h(x(t),u(t)), & 
	\end{cases}
\end{align}

where $x(t) \in X \subset R^n$, and $u(t) \in U \subset R^m$, and $y(t) \in Y \subset R^m$  are respectively the state, input and output of the system, and $X$, $U$ and $Y$ are respectively the state, input and output spaces. 
	\begin{definition}
	\label{supply}(\cite{willems1972dissipative}) The supply rate $\omega(u(t),y(t))$ is a well-defined supply rate, if for all $t_0$, $t_1$ where $t_1\geq t_0$, and all solutions $x(t)\in X$, $u(t)\in U$, and $y(t)\in Y$ of the dynamical system, we have:
		\begin{align}
			\label{eq:supp}
			\int_{t_0}^{t_1}|\omega(u(t),y(t))| dt < \infty
		\end{align}
   \end{definition}

Dissipativity and passivity are energy-based notions that characterize a dynamical system by its input/output behavior. A system is dissipative or passive if the increase in the system's stored energy is less than the entire energy supplied to it. The energy supplied to the system is defined by the supply function, and the energy stored in the system is defined by the storage function:
\begin{definition}
\label{diss}(\cite{willems1972dissipative}) System $G$ is dissipative with respect to the well-defined supply rate $\omega(u(t),y(t))$, if there exists a nonnegative storage function $V(x): X \to R^+$ such that for all $t_0$, $t_1$ where $t_1\geq t_0$, and all solutions $x(t)\in X$, $u(t)\in U$, and $y(t)\in Y$ of the dynamical system:
   \begin{align}
	\label{eq:dissipative}
	V(x(t_1))-V(x(t_0)) \leq \int_{t_0}^{t_1}\omega(u(t),y(t)) dt 
	\end{align}
is satisfied. If the storage function is differentiable, then \eqref{eq:dissipative} can be written as:
	\begin{align}
	\label{eq:dissipative2}
	\dot{V}(x(t)) \leq \omega(u(t),y(t)),~ \forall t\geq0 
	\end{align} 
\end{definition}

Accordingly, we call a system \emph{QSR-disspative} if it is dissipative with respect to the well-defined supply rate:
\begin{align}
	\label{eq:QSR}
    \omega(u(t),y(t)) = y^T(t) Q y(t) + 2 y^T(t) S u(t) + u^T(t) R u(t), 
\end{align}

where $Q$, $R$, and $S$ are constant matrices of appropriate dimensions, and $Q$ and $R$ are symmetric \cite{hill1976stability}. 
\begin{definition}(\cite{khalil2002nonlinear})
System $G$ is called finite-gain $L_2$-stable, if there exists a positive-semi definite function $V(x): X \to R^+$, and a scalar constant $\gamma>0$ such that for all $u(t)\in U$, and $y(t)\in Y$ and $t_1>0$
	\begin{align}
	 \label{eq:l2}
	V(x(t_1))-V(x(0)) \leq \int_{0}^{t_1}(\gamma ^2 u^T(t)u(t)-y^T(t)y(t)) dt 
	\end{align}
or if for the smallest possible gain $\gamma$, $\forall u(t)\in U$, a $\beta$ exists such that over the time interval $[0,\tau]$ we have:
	\begin{align*}
&||y_{\tau}||_{L2} \leq \gamma||u_{\tau}||_{L2} + \beta
	\end{align*}
\end{definition}

The relation between \emph{QSR-disspativity} and finite-gain $L_2$-stability is well-established:
\begin{theorem}\label{stab}(\cite{hill1976stability}) If system $G$ is \emph{QSR-disspative} with $Q<0$, then it is $L_2$-stable.
\end{theorem}		
\begin{definition}(\cite{bao2007process}) As a special case of dissipativity, system $G$ is called passive, if there exists a nonnegative storage function $V(x): X \to R^+$ such that:
\label{pass}
\begin{align}
	\label{eq:passivity}
	V(x(t_1))-V(x(t_0)) \leq \int_{t_0}^{t_1} u^T(t)y(t) dt 
\end{align}
is satisfied for all $t_0$, $t_1$ where $t_1\geq t_0$, and all solutions $x(t)\in X$, $u(t)\in U$, and $y(t)\in Y$ of the dynamical system. If the storage function is differentiable, then \eqref{eq:passivity} can be written as:
\begin{align}
	\label{eq:passivity2}
	\dot{V}(x(t)) \leq u^T(t)y(t), \forall t\geq0 
\end{align}
\end{definition}

Under certain conditions, passivity coincides with input/output stability, and for zero-state detectable dynamical systems, guarantees the stability of the origin \cite{khalil2002nonlinear}.
\begin{definition}(\cite{khalil2002nonlinear})\label{IFOF} System $G$ is considered to be Input Feed-Forward Output Feedback Passive (IF-OFP), if it is dissipative with respect to the following well-defined supply rate:
\label{IFOF}
\begin{align}
	\label{eq:IFOFP}
	\omega(u,y)=u^Ty-\rho y^Ty-\nu u^Tu, \forall t\geq0,
\end{align}
for some $\rho, \nu \in R$.
\end{definition}

IF-OFP property presents a more general form for the concept of passivity. Based on definition \ref{IFOF}, we can denote an IF-OFP system with IF-OFP($\nu$,$\rho$). $\nu$ is called the input passivity index and $\rho$ is called the output passivity index. Passivity indices are a means to measure the shortage and excess of passivity in dynamical systems \cite{khalil2002nonlinear}, and are useful in passivity-based analysis and control of systems \cite{yu2013event, xia2014passivity}. A positive value for either one of two passivity indices points to an excess in passivity; and a negative value for either of two passivity indices points to a shortage in passivity.  An excess of passivity in one system can compensate for the shortage of passivity in another system leading to a passive feedback or feed-forward interconnection \cite{hirche2012human}. Moreover, passivity indices can be useful for analyzing the performance of passive systems. Further, if only $\nu>0$, then the system is said to be \emph{input strictly passive} (ISP); if only $\rho>0$, then the system is said to be \emph{output strictly passive} (OSP). Similarly, if $\nu>0$ and $\rho>0$, then the system is said to be \emph{very strictly passive} (VSP). 
\begin{lemma}
(\cite{matiakis2006novel}) Without loss of generality, the domain of $\rho$ and $\nu$ in Input Feed-Forward Output Feedback Passive systems (\ref{eq:IFOFP}) is $\Pi=\Pi_1\cup\Pi_2$ with $\Pi_1=\{\rho, \nu \in R|\rho\nu<\frac{1}{4}\}$ and $\Pi_2=\{\rho, \nu \in R|\rho\nu=\frac{1}{4}\ with~ \rho\geq0\}$
\end{lemma}

\section{Problem Statement}
\label{sec:prob}
As shown in figures \ref{fig:ETNCSDPQ}, we are exploring the interconnection of two IFOF passive systems. The main plant has passivity levels $\rho_p$, and $\nu_p$, and the controller has passivity levels $\rho_c$, and $\nu_c$. We consider a large class of systems where the indices can take positive or negative values indicating the extent that each sub-system is passive or non-passive. Most linear and nonlinear systems can be described by this definition given that we know their passivity indices. Additionally, we are considering the effects of signal quantization, time-delays and information loss which are unavoidable in networked control systems. Our set-up will guarantee finite-gain $L_2$-stability and performance improvement for the networked control system given in figure \ref{fig:ETNCSDPQ} in the presence of quantization, time-delays, information loss and event-triggering conditions on both sides. This is reached based on the main concepts of dissipativity and passivity (given in definitions \ref{diss}, \ref{pass}) and the relation given in definition \ref{IFOF}. At the same time, our proposed design will greatly reduce the communication load on the shared band-limited communication network between the plant and controller.

The set-up is event-triggered on both sides. The updating process is governed by the event-detectors. This structure is commonly used in literature to analyze the behavior of networked interconnections as it can capture different NCS configurations \cite{hespanha2007survey}. The inputs to the controller or plant are held constant based on the last value received from the communication network. $w_1(t)$ can denote a reference input on the plant side or an external disturbance in cases where we are examining the robustness of our design. $w_2(t)\neq 0$ as an  external disturbance on the controller side has been only considered for examining the robustness of the triggering condition in section \ref{sec:robustness}. 

Additionally, we consider simple triggering conditions for both sides:
\begin{align}
&||e_p(t)||_2^2 > \delta_p||y_p(t)||_2^2 ~~~~ 0 < \delta_p \leq 1\\
&||e_c(t)||_2^2 > \delta_c||y_c(t)||_2^2 ~~~~ 0 < \delta_c \leq 1
\end{align}

This will facilitate the design process and also make it easier for the designer to understand and analyze the trade-offs amongst performance, finite-gain $L_2$-stability, channel utilization, and passivity levels of sub-systems and to make design decisions accordingly. Another advantage of the above conditions is that we do not need the exact dynamical models for each sub-system before making design decisions, and that an access only to each sub-system's output is sufficient. Moreover, in \cite{rahnama2016qsr} we showed the relationship amongst each of these event-triggering conditions and $\emph{QSR}-dissipativiy$ and stability, and our work in this report will utilize some of these former results as building blocks for the new findings. 

As previously mentioned, in this report we show finite-gain $L_2$-stability and robustness for the networked control set-up under the triggering conditions mentioned above and with an appropriate design of the $M$-Transformation matrix in the presence of time-varying or constant delays, signal quantizations and information loss. As you can see in the figure \ref{fig:ETNCSDPQ}, by utilizing the $M$-Transformation matrix, we will seek to transform the controller $G_c$ and the communication network into a new IFOF passive system $\tilde{G}_c$ with input and output passivity indices $\tilde{\rho}_c$ and $\tilde{\nu}_c$. Our proposed design will show a range of allowable values for entries in the $M$-Transformation matrix that guarantee finite-gain $L_2$-stability based on the triggering conditions, and passivity levels of the plant and controller. Additionally, we show a relationship amongst these values, robustness and the number of allowed consecutive packet dropouts. Hence, the designer will be able to examine these relationships and make decisions accordingly based on application, robustness, performance, reliability of communication links and efficient utilization of the shared band-limited communication channel. We will show that finite-gain $L_2$-stability, robustness and communication rate in our design depend on the passivity levels for each sub-unit, design entries in the $M$-Transformation matrix, and the flexibility of our triggering conditions, and that more passive systems with larger passivity indices have better performances, are more robust, and utilize the communication network less frequently. 
		\begin{figure}[!t]
			\centering
			\includegraphics[scale = .65]{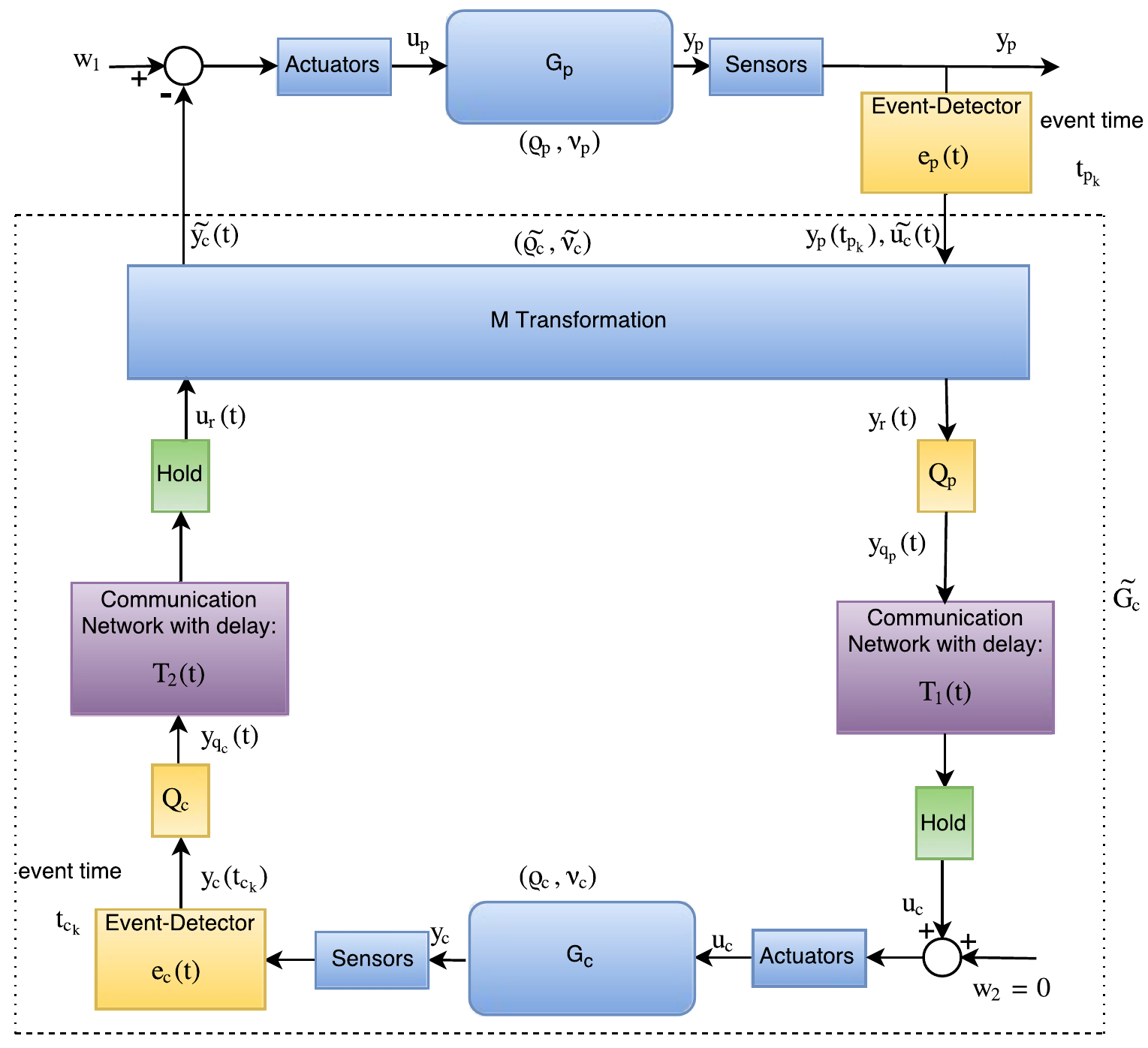}
			\caption{A Networked Control System Interconnection of two IF-OFP systems with Quantizations, Time-varying Delays and Event-detectors on both sides.}
			\label{fig:ETNCSDPQ}
		\end{figure}
\clearpage
\section{Event-Triggered Networked Control System's Structure}
	    \label{sec:structure}

\subsection{Event-Triggering Conditions}
The triggering mechanisms in the figure \ref{fig:ETNCSDPQ} are representing the situations, in which new information is sent every time a violation of the triggering condition occurs. We are considering the set-up where the IFOF passive plant $G_p$ has passivity indices $\rho_p$, $\nu_p$ and the IFOF passive controller has passivity indices $\rho_c$, $\nu_c$. An event-detector is located on the output of the plant to monitor the behavior of plant's output. An updated measure of $y_p$ is sent to the communication network when the error between the last information sent ($y_p(t_{p_{k}})$) and the current one:  $e_p(t)=y_p(t)-y_p(t_{p_{k}})$ (for $t \in [t_{p_k}, t_{p_{k+1}})$) exceeds a predetermined threshold established by the designer. A similar setup is also presented on the controller's output where the event-detector is located on the output of the controller, and an updated measure of $y_c$ is sent to the communication network when the error between the last information sent from the controller ($y_c(t_{c_{k}})$) and the current one  $e_c(t)=y_c(t)-y_c(t_{c_{k}})$ (for $t \in [t_{c_k}, t_{c_{k+1}})$) exceeds a predetermined threshold value.  The triggering conditions on the plant's and controller's sides are:
		  \begin{align}
		  \label{eq:trigp}
		  &||e_p(t)||_2^2>\delta_p||y_p(t)||_2^2~~ where~~ \delta_p  \in (0,1]
		  \end{align}
		  \begin{align}
		  \label{eq:trigc}
		  &||e_c(t)||_2^2>\delta_c||y_c(t)||_2^2~~ where~~ \delta_c  \in (0,1]
		  \end{align} 

where $e_p(t)=y_p(t)-y_p(t_{p_k})$ for all $t \in [t_{p_k},t_{p_{k+1}})$ and at instances that triggering condition is met, and new information is exchanged, the error is set back to zero: $e_p(t_{p_{k+1}})=0$. Similarly, on the controllers side, we have  $e_c(t)=y_c(t)-y_c(t_{c_k})$ for any $t \in [t_{c_k},t_{c_{k+1}})$ and $e_c(t_{c_{k+1}})=0$. This set-up results in a great decrease in the amount of information required to be exchanged between sub-systems in order for them to maintain the desirable performance. Additionally one conclusion made from triggering conditions in (\ref{eq:trigp}) and (\ref{eq:trigc}) is the following:
		 \begin{align}
		 \label{eq:trigpr}
		 &||y_p(t_{p_k})||_2 \leq (1+\sqrt{\delta_p})||y_p(t)||_2~~ for~ any ~~ t \in [t_{p_k},t_{p_{k+1}})
		 \end{align} 
		 \begin{align}
		 \label{eq:trigcr}
		 &||y_c(t_{c_k})||_2 \leq (1+\sqrt{\delta_c})||y_c(t)||_2~~ for~ any ~~ t \in [t_{c_k},t_{c_{k+1}})
		 \end{align} 
		
As mentioned earlier, in section \ref{sec:main} we show the process for determining $\delta_p$. $\delta_c$ is chosen based on our previous work on this topic given in theorem 10 in \cite{rahnama2016qsr}.

\subsection{Signal Quantizations}
$Q_p$ and $Q_c$ represent the signal quantization blocks in the framework presented in figure \ref{fig:ETNCSDPQ}. We have the following general definition for quantizers:
		\begin{align*}
		&q(v)=[q_1(v_1) q_2(v_2)...q_m(v_m)]^T, q_j(-v_j)=-q_j(v_j)
		\end{align*}
		
where $q(.), j=1,2,...,m$ is a static time-invariant symmetric quantizer. The quantization levels are expressed as the following:
		\begin{align*}
		&Q=\{\pm\sigma_j, j=\pm 1,\pm 2...\}\cup\{0\}
		\end{align*}

The quantization regions differ for various quantization approaches. We have considered a large class of quantizers called passive quantizers in our work. This class includes many different methods of quantization such as logarithmic quantizers, mid-riser and mid-tread uniform quantizers. A quantizer is passive if its input $u$ and output $y$ satisfy the following relation:
			\begin{align*}
			&au^T(t)u(t)\leq u^T(t)y(t) \leq bu^T(t)u(t)
			\end{align*}

where $y(t)=q(u(t))$ and $0\leq a \leq b<\infty$.  Moreover, if $u(t)$ is a vector then the passive quantizer functions component wise on the input vectors. Given figure \ref{fig:ETNCSDPQ}, we have:\\
			\begin{align}
			&Q_p:a_p y_r^T(t)y_r(t)\leq  y_r^T(t)y_{q_p}(t) \leq b_p y_r^T(t)y_r(t) \\
			&Q_c:a_c y_c^T(t_{c_k})y_c(t_{c_k})  \leq y_c^T(t_{c_k})y_{q_c}(t)\leq b_c y_c^T(t_{c_k})y_c(t_{c_k})
			\end{align}

where $b_c>a_c\geq0$ and $b_p>a_p\geq0$ and $y_{q_c}(t)=Q_c(y_c(t_{c_k}))$ and $y_{q_p}(t)=Q_p(y_r(t))$ denote the relation between input and output of each quantizer. Additionally, from the relationship between the input and output of a passive quantizer, we can conclude that:
            \begin{align}
         	\label{eq:trigpqr}
         	&||y_{q_p}(t)||_2^2 \leq b_p^2 ||y_r(t)||^2_2.
            \end{align} 
			\begin{align}
		    \label{eq:trigcqr}
			 &||y_{q_c}(t)||_2^2\leq b_c^2 ||y_c(t_{c_k})||_2^2
			 \end{align} 
			 
\subsection{Time-varying Network Induced Delays}
It is known that time-delays in networked control systems degrade the performance of interconnections. Most approaches in the literature can be divided into two categories: delay-dependent and delay-independent methods \cite{richard2003time,gu2003stability}. In our work, we propose a time-independent input-output stability approach, where we do not assume that the length of time-delays are known.  $T_1(t)$ represents the time-varying network induced delay from the plant to controller. $T_2(t)$ represents the time-varying network induced delay from the controller to plant. We assume that there is no bound on admissible network induced delays, namely, they may be larger than inter-event time intervals. However, we have considered a causality assumption and a bounded rate of change for time-delays meaning:
		\begin{align}
		\label{eq:timed1}	
	     &0\leq |\frac{dT_1(t)}{dt}| \leq d_1<1
		\end{align} 
		\begin{align}
		\label{eq:timed2}	
		 &0\leq |\frac{dT_2(t)}{dt}| \leq d_2<1.
		\end{align} 
	
These conditions simply mean that the time-delays cannot grow faster than time itself. The hold blocks function as zero-order holders, meaning at each triggering instance new information is sent from the plant or controller to the hold over the communication network, and in between these instances, the hold blocks will have constant values  that are set based on the last information they have received.
\subsection{$M$-Transformation Matrix}
The idea of implementing a matrix transformation was adopted from the research done in the field of communication theory where similar methods are used to overcome the negative effects of time-delays in passive interconnections \cite{sontag1989smooth,lozano2002passivation}, and similarly in control theory to preserve passivity \cite{hirche2009distributed} or to passivate non-passive systems \cite{xia2014passivity,Arash2016,xia2015guaranteeing}. However, authors in these works do not consider quantization and time-delays together, they assume constant time-delays or in the case where they are dealing with networked control systems, they consider continuous or periodic information exchange between sub-systems. The $M$-Transformation matrix on the plant's side of communication links (figure \ref{fig:Mmatrix}) can be looked at as a simple local controller consisting of only simple gains such that:
		\begin{align}
		\label{eq:M}
		&\begin{bmatrix}
		y_r(t)\\u_r(t)
		\end{bmatrix}=M \begin{bmatrix}
		\tilde{u}_c(t)\\\tilde{y}_c(t)
		\end{bmatrix}=\begin{bmatrix}
		m_{11}I_m& 0I_m\\
		m_{21}I_m & m_{22}I_m\end{bmatrix} \begin{bmatrix}
		\tilde{u}_c(t)\\\tilde{y}_c(t)
		\end{bmatrix}
		\end{align}

where $I_m \in R^{m\times m}$. The presence of $M$-Transformation matrix gives the networked control interconnection a distributed control architecture where the complicated control actions and computations are being processed at the remote controllers. Our intention is to design the $M$-Transformation matrix such that the negative effects of quantizations on both sides, and time delays on both sides are overcome, and as a result the entire system is both finite-gain $L_2$-stable and robust. 
	
		\begin{figure}[!t]
				\centering
				\includegraphics[scale = 1]{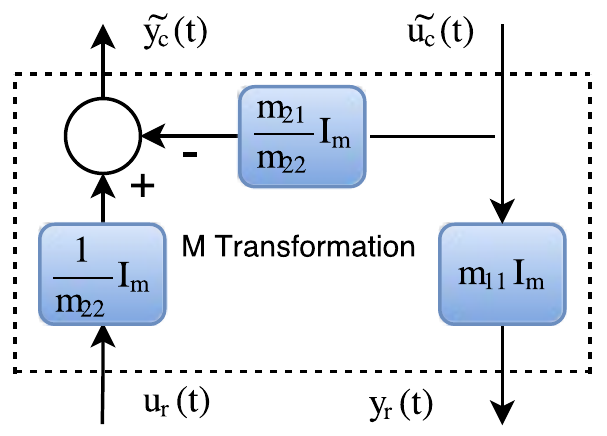}
				\caption{Proposed Transformation for overcoming the  effects of time-varying delays and signal quantization ($M$-Transformation Matrix).}
				\label{fig:Mmatrix}
		\end{figure}
\clearpage		
\section{Main Results}
\label{sec:main}	
\subsection{Finite-Gain $L_2$-Stability Analysis of Triggering Conditions}
\begin{theorem}
Consider the networked interconnection of two IFOF passive systems $G_p$ and $\tilde{G}_c$ in figure \ref{fig:ETNCSDPQ} with respective passivity indices of $\nu_p$, $\rho_p$, $\tilde{\nu}_c$ and $\tilde{\rho}_c$. Let the event instance $t_{p_k}$ be explicitly determined by the triggering condition $||e_p(t)||_2^2>\delta_p||y_p(t)||_2^2$ where $\delta_p  \in(0,1]$, and  $\beta(\tilde{\nu}_c)>\frac{1}{4\gamma}$ and $\tilde{\rho}_c+\nu_p-|\nu_p|-\frac{1}{2\alpha}>0$, where $\alpha,\gamma>0$ and:\\
			
			\[\beta(\tilde{\nu}_c) =
			\begin{cases}
			\rho_p-\frac{\delta_p\alpha}{2}      & \quad \text{if } \tilde{\nu}_c \geq 0\\
			\rho_p+2\tilde{\nu}_c-\delta_p(\frac{\alpha}{2}-2\tilde{\nu}_c)  & \quad \text{if } \tilde{\nu}_c < 0\\
			\end{cases}
			\]
Then the event-triggered networked control system is finite-gain $L_2$-stable from the input $
			w_1(t)$ to output
			$y_p(t)$. 		
\end{theorem}
\begin{proof} 
\label{thm:NCSYpEp}
Given the systems $G_p$ with passivity indices $\nu_p$, $\rho_p$, and $\tilde{G}_c$ which is representing the transformed system encompassing the controller, network connections and the $M$-Transformation matrix as given in figure \ref{fig:ETNCSDPQ} with passivity indices  $\tilde{\nu}_c$ and $\tilde{\rho}_c$, there exist $V_p(t)$ and $\tilde{V}_c(t)$ such that:
	
			\begin{align*}
			& \dot{V}_p(t) \leq u^T_p(t)y_p(t)-\nu_pu^T_p(t)u_p(t)-\rho_p y^T_p(t)y_p(t) 
			\\&\dot{\tilde{V}}_c(t) \leq \tilde{u}^T_c(t)\tilde{y}_c(t)-\tilde{\nu}_c \tilde{u}^T_c(t)\tilde{u}_c(t)-\tilde{\rho}_c \tilde{y}^T_c(t)\tilde{y}_c(t).\\
			\end{align*}

Additionally, according to the setup portrayed in figure \ref{fig:ETNCSDPQ}, the following relationships stand for any $t \in [t_{p_k},t_{p_{k+1}})$:
			
			\begin{align*}
			& u_p(t)=w_1(t)-\tilde{y}_c(t)
			\\&e_p(t)=y_p(t)-y_p(t_{p_k})
			\\&\tilde{u}_c(t)=y_p(t_{p_k})=y_p(t)-e_p(t).
			\end{align*}

We design the triggering condition based on the following rule ($||e_p(t)||_2^2>\delta_p||y_p(t)||_2^2 $):	
			\begin{align*}
			& \langle e_p,e_p\rangle>\delta_p \langle y_p,y_p\rangle, ~ 0 < \delta_p \leq 1.
			\end{align*} 
			
We consider the following storage function for the interconnection:
			
			\begin{align*}
			& V(t) = V_p(t) + \tilde{V}_c(t),
			\end{align*} 
			
as a results, we have:
			\begin{align*}
			&\dot{V}(t)=\dot{V}_p(t)+\dot{\tilde{V}}_c(t)\\&~~~~~~ \leq u^T_p(t)y_p(t)-\nu_pu^T_p(t)u_p(t)-\rho_p y^T_p(t)y_p(t)+\tilde{u}^T_c(t)\tilde{y}_c(t)-\tilde{\nu}_c \tilde{u}^T_c(t)\tilde{u}_c(t)-\tilde{\rho}_c \tilde{y}^T_c(t)\tilde{y}_c(t).
			\end{align*} 
			
We know that $u_p(t)=w_1(t)-\tilde{y}_c(t)$, $\tilde{u}_c(t)=y_p(t)-e_p(t)$, consequently for any $t \in [t_{p_k},t_{p_{k+1}})$ we have:
			
			\begin{align*}
			&\dot{V}(t) \leq (w_1(t)-\tilde{y}_c(t))^T y_p(t)-\nu_p(w_1(t)-\tilde{y}_c(t))^T(w_1(t)-\tilde{y}_c(t))-\rho_p y^T_p(t)y_p(t)
			\\&~~~~~~+(y_p(t)-e_p(t))^T\tilde{y}_c(t)-\tilde{\nu}_c (y_p(t)-e_p(t))^T(y_p(t)-e_p(t))-\tilde{\rho}_c \tilde{y}^T_c(t)\tilde{y}_c(t)
			\\&~~~~~~
			=w_1^T(t)y_p(t)-\nu_pw_1^T(t)w_1(t)-(\rho_p+\tilde{\nu}_c)y^T_p(t)y_p(t)-(\tilde{\rho}_c+\nu_p)\tilde{y}^T_c(t)\tilde{y}_c(t)
			\\&~~~~~~+2\nu_pw^T_1(t)\tilde{y}_c(t)+2\tilde{\nu}_cy^T_p(t)e_p(t)-e_p^T(t)\tilde{y}_c(t)-\tilde{\nu}_ce_p^T(t)e_p(t).
			\end{align*} 
			
			 We can show: $-e_p^T(t)\tilde{y}_c(t)=-(\frac{\sqrt{\alpha}e_p(t)}{\sqrt{2}}+\frac{\tilde{y}_c(t)}{\sqrt{2\alpha}})^2+\frac{\alpha}{2} e_p^T(t)e_p(t)+\frac{1}{2\alpha}\tilde{y}_c^T(t)\tilde{y}_c(t)$ where $\alpha>0$, and $w_1^T(t)y_p(t)=-(\sqrt{\gamma} w_1(t)-\frac{y_p(t)}{2\sqrt{\gamma}})^2+\gamma w^T_1(t)w_1(t)+\frac{y^T_p(t)y_p(t)}{4\gamma}$ where $\gamma>0$. Additionally, we have $e_p^T(t)e_p(t) \leq \delta_p y_p^T(t)y_p(t)$ for all $t \in [t_{p_k},t_{p_{k+1}})$ and $2\nu_pw^T_1(t)\tilde{y}_c(t)\leq |\nu_p|w^T_1(t)w_1(t)+|\nu_p|\tilde{y}_c^T(t)\tilde{y}_c(t)$.
			  Hence, simplifying further, we can have: 
			
			$\dot{V}(t) \leq \gamma w^T_1(t)w_1(t)+\frac{y^T_p(t)y_p(t)}{4\gamma}-\nu_p w^T_1(t)w_1(t)  
			+|\nu_p|w^T_1(t)w_1(t)+|\nu_p|\tilde{y}_c^T(t)\tilde{y}_c(t)\\~~~~~~~~~~~~-(\rho_p+\tilde{\nu}_c-\frac{\delta_p\alpha}{2}) y^T_p(t)y_p(t)-(\tilde{\rho}_c+\nu_p-\frac{1}{2\alpha})\tilde{y}_c^T(t)\tilde{y}_c(t) 
			\\~~~~~~~~~~~~+2\tilde{\nu}_cy^T_p(t)e_p(t)-\tilde{\nu}_ce_p^T(t)e_p(t)$,\\
			
if $\tilde{\nu_c} \geq 0$, we can utilize the following relationship $2\tilde{\nu}_cy^T_p(t)e_p(t) \leq |\tilde{\nu}_c|y_p^T(t)y_p(t)+|\tilde{\nu}_c|e_p^T(t)e_p(t)$ back into the above equation to have:\\ 
			
			$\dot{V}(t) \leq \gamma w^T_1(t)w_1(t)+\frac{y^T_p(t)y_p(t)}{4\gamma}-\nu_p w^T_1(t)w_1(t)  
			+|\nu_p|w^T_1(t)w_1(t)\\~~~~~~~~~~~~-(\rho_p-\frac{\delta_p\alpha}{2}) y^T_p(t)y_p(t)-(\tilde{\rho}_c+\nu_p-|\nu_p|-\frac{1}{2\alpha})\tilde{y}_c^T(t)\tilde{y}_c(t)$,\\
			
if $\tilde{\nu}_c < 0$, given that $2\tilde{\nu}_cy^T_p(t)e_p(t) \leq |\tilde{\nu}_c|y_p^T(t)y_p(t)+|\tilde{\nu}_c|e_p^T(t)e_p(t)$, and $|\tilde{\nu}_c|e_p^T(t)e_p(t) \leq \delta_p|\tilde{\nu}_c| y_p^T(t)y_p(t)$ and $-\tilde{\nu}_c e_p^T(t)e_p(t) \leq \delta_p |\tilde{\nu}_c| y_p^T(t)y_p(t)$ we have: \\
			
			$\dot{V}(t) \leq \gamma w^T_1(t)w_1(t)+\frac{y^T_p(t)y_p(t)}{4\gamma}-\nu_p w^T_1(t)w_1(t)  
			+|\nu_p|w^T_1(t)w_1(t)\\~~~~~~~~~~~~-(\rho_p+2\tilde{\nu}_c-\delta_p(\frac{\alpha}{2}-2\tilde{\nu}_c)) y^T_p(t)y_p(t)-(\tilde{\rho}_c+\nu_p-|\nu_p|-\frac{1}{2\alpha})\tilde{y}_c^T(t)\tilde{y}_c(t)$.\\
			
Given $\tilde{\rho}_c+\nu_p-|\nu_p|-\frac{1}{2\alpha}>0$, we have:\\
			
			$\dot{V}(t) \leq (\gamma+|\nu_p|-\nu_p) w^T_1(t)w_1(t)  - (\beta(\tilde{\nu}_c)-\frac{1}{4\gamma})y^T_p(t)y_p(t)$,
		
where:
			
			\[\beta(\tilde{\nu}_c) =
			\begin{cases}
			\rho_p-\frac{\delta_p\alpha}{2}      & \quad \text{if } \tilde{\nu}_c \geq 0\\
			\rho_p+2\tilde{\nu}_c-\delta_p(\frac{\alpha}{2}-2\tilde{\nu}_c)  & \quad \text{if } \tilde{\nu}_c < 0\\
			\end{cases}
			\]
			
We know $\beta(\tilde{\nu}_c)>\frac{1}{4\gamma}$, and we integrate from both sides over the time interval $[0,\tau]$:
			
				\begin{align}
					&(\beta(\tilde{\nu}_c)-\frac{1}{4\gamma})\int_{0}^{\tau}y^T_p(t)y_p(t)dt+\int_{0}^{\tau}\dot{V}(t)dt\leq (\gamma+|\nu_p|-\nu_p)\int_{0}^{\tau}w^T_1(t)w_1(t)dt
				\end{align}
			
Given $V(x)\geq0$, we take the square root and simplify to have:
			\begin{align}
			&||y_{p_\tau}||_{L2} \leq (\frac{\gamma+|\nu_p|-\nu_p}{\beta(\tilde{\nu}_c)-\frac{1}{4\gamma}})||w_{1_\tau}||_{L2} + \sqrt{\frac{V(0)}{\beta(\tilde{\nu}_c)-\frac{1}{4\gamma}}}
			\end{align}
			
which leads to the definition of finite-gain $L_2$-stability and proves the theorem. It is important to note that $y_{p_\tau}$ and $w_{1_\tau}$ represent the integral terms including signals $y_p(t)$ and $w_1(t)$ over the time interval  $[0,\tau]$. 
\end{proof}

\begin{remark}
The proof reveals an interesting relationship between passivity indices and finite-gain $L_2$-stability: larger values of $\beta(\tilde{\nu}_c)$ lead to tighter upper-bounds on the output. Relating this to passivity indices in $\beta(\tilde{\nu}_c)$, one finds out that a larger output passivity index $\rho_p$ leads to a smaller $L_2$ gain or a tighter upper-bound for the output. Similarly, there is a direct relation between plant's input passivity index and the interconnection's $L_2$ gain: a positive $\nu_p$ leads to a tighter upper-bound on the output compared to a negative input passivity index.
\end{remark}

\subsection{The Design for the $M$-Transformation Matrix}
\begin{theorem}
Consider the event-triggered networked control system shown in figure \ref{fig:ETNCSDPQ}. Let the quantizers be passive quantizers, $\rho_c>0$, and network induced time-varying delays from the plant to controller $(T_1(t))$ and the controller to plant $(T_2(t))$ meet the causality conditions:\\
			\begin{align*}
			&0\leq |\frac{dT_1(t)}{dt}| \leq d_1<1 \\& 0\leq |\frac{dT_2(t)}{dt}| \leq d_2<1
			\end{align*}
			let the triggering instances on the plant's, and controller's sides be explicitly determined by the following relations:
			\begin{align*}
			&||e_p(t)||_2^2>\delta_p||y_p(t)||_2^2~~ where~~ \delta_p  \in (0,1]\\
			&||e_c(t)||_2^2>\delta_c||y_c(t)||_2^2~~ where~~ \delta_c  \in (0,1]
			\end{align*} 
			if the entries of the $M$-Transformation matrix are selected such that:   
			\begin{align*}
			&m_{22}^2 > \frac{(\frac{1}{2\alpha}+|\nu_p|-\nu_p)(2b_c^2(1+\sqrt{\delta_c})^2(1+d_2))}{\rho_c}~~~where~~\tilde{\rho}_c=\frac{\rho_cm_{22}^2}{2b_c^2(1+\sqrt{\delta_c})^2(1+d_2)}\\&
			|m_{21}|=\frac{b_c^2(1+\sqrt{\delta_c})^2(1+d_2)}{\rho_c|m_{22}|} ~~~ where~~~  m_{21}m_{22}<0\\
			&\tilde{\nu}_c=\frac{b_c^2(1+\sqrt{\delta_c})^2(1+d_2)}{2\rho_c m_{22}^2}-(\frac{1}{2\rho_c}+|\nu_c|)b_p^2(1+d_1)m_{11}^2~~~~where~~ \tilde{\rho}_c\tilde{\nu}_c<\frac{1}{4}
			\end{align*}
		    to meet the conditions given in theorem 2, namely
			$\beta(\tilde{\nu}_c)>\frac{1}{4\gamma}$ and $\tilde{\rho}_c+\nu_p-|\nu_p|-\frac{1}{2\alpha}>0$, where $\alpha,\gamma>0$ and:\\
			\[\beta(\tilde{\nu}_c) =
			\begin{cases}
			\rho_p-\frac{\delta_p\alpha}{2}      & \quad \text{if } \tilde{\nu}_c \geq 0\\
			\rho_p+2\tilde{\nu}_c-\delta_p(\frac{\alpha}{2}-2\tilde{\nu}_c)  & \quad \text{if } \tilde{\nu}_c < 0\\
			\end{cases}
			\]
			then the networked control system is finite-gain $L_2$-stable from the input $w_1(t)$ to output $y_p(t)$. 
		\end{theorem}
		\begin{proof}
			Given $\rho_c > 0$, for the IFOF passive controller we have:
			\begin{align}
			&\dot{V}_c(t) \leq u^T_c(t)y_c(t)-\nu_c u^T_c(t)u_c(t)-\rho_c y^T_c(t)y_c(t)\\
			&~~~~~~\leq |\nu_c|u^T_c(t)u_c(t)-\rho_c y^T_c(t)y_c(t) +u^T_c(t)y_c(t)\\
			&~~~~~~\leq-\frac{1}{2\rho_c}(u_c(t)-\rho_cy_c(t))^2+\frac{1}{2\rho_c}u^T_c(t)u_c(t)-\frac{\rho_c}{2}y^T_c(t)y_c(t)+|\nu_c|u^T_c(t)u_c(t)\\
			&~~~~~~\leq(\frac{1}{2\rho_c}+|\nu_c|)u^T_c(t)u_c(t)-\frac{\rho_c}{2} y^T_c(t)y_c(t)\\
			&~~~~~~\leq(\frac{1}{2\rho_c}+|\nu_c|)||u_c(t)||^2_2-\frac{\rho_c}{2}||y_c(t)||^2_2.
			\end{align}	
			
			Integrating from both sides over the time interval $t_0$ to $t$ ($\forall t \geq t_0 \geq 0$) we have:
			\begin{align*}
			&\Delta V_c=V_c(t)-V_c(t_0)\leq (\frac{1}{2\rho_c}+|\nu_c|) \int_{t_0}^{t}||u_c(\tau)||^2_2 d\tau -\frac{\rho_c}{2}\int_{t_0}^{t}||y_c(\tau)||^2_2 d\tau.
			\end{align*} 		
			
			Additionally, we have: 	
			\begin{align}
			&\int_{t_0}^{t}||u_c(\tau)||^2_2 d\tau=\sum_{k=0}^{N}[t_{p_{k+1}}-t_{p_{k}}+T_1(t_{p_{k+1}})-T_1(t_{p_{k}})]y_{q_p}(t_{p_k})~~~~~~~      (ZOH)\\
			&~~~~~~~~~~~~~~~~~~~~\leq\sum_{k=0}^{N}[(1+d_1)(t_{p_{k+1}}-t_{p_{k}})]y_{q_p}(t_{p_k}))~~~~~~~~~~~~~~~~~      (Given~~ \ref{eq:timed1})\\
			&~~~~~~~~~~~~~~~~~~~~\leq b_p^2(1+d_1)\int_{t_0}^{t}||y_r(\tau)||^2_2 d\tau~~~~~~~~~~~~~~~~~~~~~~~~~~~~      (Given~~ \ref{eq:trigpqr})
			\end{align}
			
			Similarly, we have: 
			\begin{align}
			&\int_{t_0}^{t}||u_r(\tau)||^2_2 d\tau=\sum_{k=0}^{N}[t_{c_{k+1}}-t_{c_{k}}+T_2(t_{c_{k+1}})-T_2(t_{c_{k}})]y_{q_c}(t_{c_k})~~~~~~~  ~    (ZOH)\\
			&~~~~~~~~~~~~~~~~~~~~\leq\sum_{k=0}^{N}[(1+d_2)(t_{c_{k+1}}-t_{c_{k}})]y_{q_c}(t_{c_k})~~~~~~~~~~~~~~~~~~      (Given~~ \ref{eq:timed2})\\
			&~~~~~~~~~~~~~~~~~~~~\leq (1+d_2)\int_{t_0}^{t}||y_{q_c}(\tau)||^2_2 d\tau 
			\end{align}	
			
			Given (\ref{eq:trigc}), and (\ref{eq:trigcqr}), we have: 
			\begin{align}
			&||y_{q_c}(t_{c_k})||^2_2 \leq  b_c^2||y_c(t_{c_k})||^2_2 \\
			&||y_{q_c}(t_{c_k})||^2_2 \leq  b_c^2(1+\sqrt{\delta_c})^2||y_c(t)||^2_2. 
			\end{align}
			
			Based on (35) and (37), we have:	
			\begin{align}
			&\int_{t_0}^{t}||u_r(\tau)||^2_2 d\tau\leq b_c^2(1+\sqrt{\delta_c})^2(1+d_2)\int_{t_0}^{t}||y_c(\tau)||^2_2 d\tau
			\end{align}
			
			Given (29), (32), and (38) we have: 
			\begin{align}
			&\dot{V}_c(t) \leq (\frac{1}{2\rho_c}+|\nu_c|)||u_c(t)||^2_2-\frac{\rho_c}{2}||y_c(t)||^2_2\\&~\Delta V_c\leq (\frac{1}{2\rho_c}+|\nu_c|)b_p^2(1+d_1)\int_{t_0}^{t}||y_r(\tau)||^2_2 d\tau-\frac{\rho_c}{2b_c^2(1+\sqrt{\delta_c})^2(1+d_2)}\int_{t_0}^{t}||u_r(\tau)||^2_2 d\tau
			\end{align}   
			
			Additionally based on the relation given in (\ref{eq:M}), we have:
			\begin{align*}
			&y_r(t)=m_{11}\tilde{u}_c(t)\\
			&u_r(t)=m_{21}\tilde{u}_c(t)+m_{22}\tilde{y}_c(t).
			\end{align*}
			
			Leading to:
			\begin{align}
			&~ \Delta V_c\leq (\frac{1}{2\rho_c}+|\nu_c|)b_p^2(1+d_1)\int_{t_0}^{t}||y_r(\tau)||^2_2 d\tau-\frac{\rho_c}{2b_c^2(1+\sqrt{\delta_c})^2(1+d_2)}\int_{t_0}^{t}||u_r(\tau)||^2_2 d\tau
			\\&~~~~~~~\leq-\frac{\rho_cm_{21}m_{22}}{b_c^2(1+\sqrt{\delta_c})^2(1+d_2)}\int_{t_0}^{t}\tilde{y}_c^T(\tau)\tilde{u}_c(\tau) d\tau\\&~~~~~~~+[(\frac{1}{2\rho_c}+|\nu_c|)b_p^2(1+d_1)m_{11}^2-\frac{\rho_cm_{21}^2}{2b_c^2(1+\sqrt{\delta_c})^2(1+d_2)}]\int_{t_0}^{t}\tilde{u}_c^T(\tau)\tilde{u}_c(\tau) d\tau \\&~~~~~~~-\frac{\rho_cm_{22}^2}{2b_c^2(1+\sqrt{\delta_c})^2(1+d_2)}\int_{t_0}^{t}\tilde{y}_c^T(\tau)\tilde{y}_c(\tau) d\tau
			\end{align}   
			
			From the previous sub-section 5.1, we know that in order for the networked control system to be finite-gain $L_2$-stable, we need $\tilde{\rho}_c>\frac{1}{2\alpha}+|\nu_p|-\nu_p$. From (44) we find:
			\begin{align}
			&m_{22}^2 > \frac{(\frac{1}{2\alpha}+|\nu_p|-\nu_p)(2b_c^2(1+\sqrt{\delta_c})^2(1+d_2))}{\rho_c}
			\end{align}
			
			Similarly,  we need the following to hold:
			\begin{align}
			&-\frac{\rho_cm_{21}m_{22}}{b_c^2(1+\sqrt{\delta_c})^2(1+d_2)}=1.
			\end{align}
			
			As a result, we find:
			\begin{align}
			&|m_{21}|=\frac{b_c^2(1+\sqrt{\delta_c})^2(1+d_2)}{\rho_c|m_{22}|},
			\end{align}
			
			where
			\begin{align*}
			&m_{21}m_{22}<0.
			\end{align*}
			
			 Now we can calculate $\tilde{\nu}_c$:
			\begin{align}
			&\frac{b_c^2(1+\sqrt{\delta_c})^2(1+d_2)}{2\rho_c m_{22}^2}-(\frac{1}{2\rho_c}+|\nu_c|)b_p^2(1+d_1)m_{11}^2=\tilde{\nu}_c~~~~where~~ \tilde{\rho}_c\tilde{\nu}_c<\frac{1}{4},
			\end{align}
			
			which completes the proof - given that we want $\beta(\tilde{\nu}_c)>\frac{1}{4\gamma}$, for a positive $\tilde{\nu}_c$, we will need to choose $m_{11}$ so that $\tilde{\nu}_c>0$.
		\end{proof}
		
		\begin{remark}
			As a design process, one can first select $m_{22}$ to meet the required $L_2$-stability condition $\tilde{\rho}_c>\frac{1}{2\alpha}+|\nu_p|-\nu_p$ and desired performance and robustness index. Consequently, $m_{21}$ is easily calculated from $m_{22}$. And as the last step, we can calculate $m_{11}$ from (48) for the desired $\tilde{\nu}_c$.  
		\end{remark}
		\begin{remark}
			In (48), one can see that $m_{22}$ and $m_{11}$ have a reciprocal relation, and that increasing $m_{22}$ will decrease $m_{11}$ for a positive $\tilde{\nu}_c$. This is expected due to the relationship between input passivity index and output passivity index in IFOF passive systems $(\tilde{\rho}_c\tilde{\nu}_c<\frac{1}{4})$.  
		\end{remark}
			\begin{remark}
				Our work considers the delays from the plant to controller, and from the controller to plant. We require no upper-bound for the delays. However, we have assumed that the time-delays are continuously differentiable. The conditions:
				\begin{align*}
				&0\leq |\frac{dT_1(t)}{dt}| \leq d_1<1 \\& 0\leq |\frac{dT_2(t)}{dt}| \leq d_2<1
				\end{align*} 
				simply mean that the time-delays cannot grow faster than time itself, hence they are statements about casualty of the networked control interconnection.  
			\end{remark}
		\begin{remark}
			Another way of looking at this theorem is that our specific selection of $M$ matrix will result in an IFOF passive system $\tilde{G}_c$ with passivity indices $\tilde{\rho}_c$ and $\tilde{\nu}_c$, which is in a feedback interconnection with $G_p$. And as long as the conditions in theorem 2 are met, the entire system is finite-gain $L_2$-stable. 
		\end{remark}
		\begin{remark}
			Theorem 3 and the triggering conditions given in theorem 2 leave a lot of room for the designer to select parameters such that the desired communication rate, robustness and performance criteria are met. The selection of $\delta_p$ and $\delta_c$ has a greater effect on the rate of communication needed by the networked control system to maintain stability. The selection of the entries of the transformation matrix, by determining the passivity indices $\tilde{\rho}_c$ and $\tilde{\nu}_c$, has a greater effect on the performance and robustness of the set-up. 
		\end{remark}
		\begin{remark}
			Our results show that the selection of the triggering condition on the plant's side depends on the output passivity index of the plant. Values of $\rho_p$ and $\delta_p$ together will determine the stability criteria of the the networked control interconnection. On the other hand, $\delta_c$ can be chosen relatively large in compared to $\delta_p$ to decrease the communication rate from the controller to plant, however larger values of $\delta_c$ will require larger gain blocks in the local controller on the plant's side; this can be interpreted as cost increase. This means that the selection of $\delta_c$ has a direct effect on the communication rate and cost.
		\end{remark}
		\begin{remark}
			 Our design needs a local controller at the plant side leading to a decentralized control platform. But as illustrated in figure \ref{fig:Mmatrix}, the local controller only requires a direct output feedback loop from $\tilde{u}_c$ to $\tilde{y}_c$ with a constant gain. Our proposed set-up like other existing networked control frameworks will still require a remote implementation for more complex control and computation tasks such as optimization (optimal control), adaptive control and so on. 
		\end{remark}
	
		\section{Robustness of the Proposed Networked Control Design}
		\label{sec:robustness}
		In this section, we seek to examine the robustness of our design in two categories: robustness against external disturbance, and robustness against packet dropouts and information loss: 
		
		In sections 6.1.1 and 6.1.2, we show lower-bounds for the time intervals between triggering instances for the event-triggering condition on the plant's side, and the event-triggering condition on the controller's side. Based on our findings in these two sub-sections, we can characterize a relationship between our proposed networked control design's performance and undesirable external disturbances. We can define this relationship as a measure of our design's robustness against uncertainties.   
		
		In sections 6.2.1 and 6.2.2, we show a trade-off between passivity levels of sub-systems in networked control interconnection, selected triggering conditions, the entries in the proposed $M$-Transformation matrix and the maximum number of allowable consecutive packet losses in communication links from the plant to controller, and from the controller to plant. The problem of information loss and packet dropouts arises  when unreliable communication channels such as wireless networks or general-purpose channels are implemented in the networked control framework. In these sections, we seek to characterize a relationship between design parameters and network interconnection's robustness against packet dropouts. Consequently, one can select design parameters according to the desired communication rate and reliability of the network connections. 
		\subsection{Analysis of Inter-Event Time Intervals (Zeno-Behavior) for Event-Triggering Conditions}
		\subsubsection{Inter-Event Time-Interval Analysis of the Event-Triggering condition on the Plant's Output}
		\begin{proposition}	
			Consider the event-triggered networked control system given in figure \ref{fig:ETNCSDPQ} with time-delays $T_1(t)$ and $T_2(t)$, and quantization blocks $Q_p$ and $Q_c$. Let the plant and controller be Input Feed-forward Output Feedback passive with respective passivity indices $\rho_p$, $\nu_p$, and $\rho_c$, $\nu_c$. Consider also that the transformation matrix $M$ is implemented according to theorem 3.	Let the triggering instances on the plant's, and controller's sides be explicitly determined by the following relations:
			\begin{align*}
			&||e_p(t)||_2^2>\delta_p||y_p(t)||_2^2~~ where~~ \delta_p  \in (0,1]\\
			&||e_c(t)||_2^2>\delta_c||y_c(t)||_2^2~~ where~~ \delta_c  \in (0,1]
			\end{align*} 
			Assuming that the input $w_1$ meets the following conditions:\\
			1) $||\dot{w}_1(t)||_2\leq C_0$ for $t \in [t_{p_k},t_{p_{k+1}})$\\
			2) $sup_{t \in [t_{p_k},t_{p_{k+1}})}||w_1(t)||_2\leq C_1$.\\
			Then for any initial condition  $x_p(0)$ for the plant, the inter-event time interval $\{t_{p_{k+1}} - t_{p_k}\}$ is lower bounded by:
			\begin{align}
			&t_{p_{k+1}}^- - t_{p_k} \geq \frac{\sqrt{\delta_p}||y_p(t_{p_{k+1}}^-))||_2}{\frac{C_0}{\rho_p}+\Gamma(2\theta_{k,p})(\frac{1}{\rho_p^2}+1)(C_1+C_2)}
			\end{align}
			where 	
			\begin{align}
			&\Gamma(2\theta_{k,p})=||cos^{-1}(\frac{\nu_p+\rho_p}{\sqrt{(1-4\rho_p\nu_p)+(\nu_p+\rho_p)^2}})||_2
			\end{align}
		\end{proposition}
		\begin{proof}
			Our proof will be based on the conic properties of IFOF passive systems \cite{kottenstette2009digital,mccourt2009connection,bridgemanextended}. It is beneficial to point out that  the input and output of system $G_p$ in figure \ref{fig:ETNCSDPQ} with  passivity indices of $\nu_p$, $\rho_p$ occupy a conic sector in the input-output space \cite{zames1966input,zames1966input2,hirche2009distributed} which is defined by its center line $\theta_z$ and its apex angle $2\theta_{k,p}$. More importantly, at any instance $t$, the input and output lie within the conic sector:
			\begin{align}
			&\theta_p(t) \in [\theta_z-\theta_{k,p},\theta_z+\theta_{k,p}]
			\end{align}
			
			Additionally, we can define the relationship between the input and output with $\theta_p(t)$ by parameterizing them in polar coordinates:
			
			\begin{align}
			&u_p(t)=r_p(t)cos(\theta_p(t))\\
			&y_p(t)=r_p(t)sin(\theta_p(t)).
			\end{align}
			
			Furthermore, we have:
			\begin{align}
			&cot(2\theta_z)=\nu_p-\rho_p,~~ \theta_z \in [0, \frac{\pi}{2}]\\
			&cos(2\theta_{k,p})=\frac{\nu_p+\rho_p}{\sqrt{(1-4\rho_p\nu_p)+(\nu_p+\rho_p)^2}},~~\theta_{k,p} \in [0,\frac{\pi}{2}).
			\end{align}
			
			Since $e_p(t)=y_p(t)-y_p(t_{p_k})$ for all $t \in [t_{p_k},t_{p_{k+1}})$, we have the following for any $t \in [t_{p_k},t_{p_{k+1}})$:
			\begin{align}
			&\frac{d}{dt}||e_p(t)||_2\leq||\dot{e}_p(t)||_2=||\dot{y}_p(t)||_2=||\dot{r}_p(t)sin(\theta_p(t))||_2+||r_p(t)\dot{\theta}_p(t)cos(\theta_p(t))||_2.
			\end{align}
			
			Given (51), we have $||\dot{\theta}_p(t)||_2\leq||2\theta_{k,p}||_2$, where:
			\begin{align}
			&||2\theta_{k,p}||_2=||cos^{-1}(\frac{\nu_p+\rho_p}{\sqrt{(1-4\rho_p\nu_p)+(\nu_p+\rho_p)^2}})||_2=\Gamma(2\theta_{k,p})
			\end{align}
			
			Hence we have:
			\begin{align}
			&\frac{d}{dt}||e_p(t)||_2\leq||\dot{e}_p(t)||_2=||\dot{y}_p(t)||_2=||\dot{r}_p(t)sin(\theta_p(t))||_2+||r_p(t)\dot{\theta}_p(t)cos(\theta_p(t))||_2\\
			&~~~~~~~~~~~~~~\leq||\dot{r}_p(t)sin(\theta_p(t))||_2+\Gamma(2\theta_{k,p})||u_p(t)||_2.
			\end{align}
			
			Furthermore, we have:
			\begin{align}
			&||\dot{r}_p(t)||_2=||\frac{d}{dt}(\frac{u_p(t)}{cos(\theta_p(t))})||_2=||\frac{\dot{u}_p(t)cos(\theta_p(t))+u_p(t)\dot{\theta}_p(t)sin(\theta_p(t))}{cos^2(\theta_p(t))}||_2
			\end{align}	
			
			Inserting the above in (59), we have:
			\begin{align}
			&\frac{d}{dt}||e_p(t)||_2\leq||\dot{u}_p(t)tan(\theta_p(t))+u_p(t)\Gamma(2\theta_{k,p})tan^2(\theta_p(t))||_2+\Gamma(2\theta_{k,p})||u_p(t)||_2
			\end{align}
			
			Since we need to have $\rho_p>\frac{\delta_p\alpha}{2}+\frac{1}{4\gamma}$ or $\rho_p>\delta_p(\frac{\alpha}{2}-2\tilde{\nu}_c)-2\tilde{\nu}_c+\frac{1}{4\gamma}$ with $\gamma>0$ for the interconnection to be $L_2$-stable, we have $\frac{y_p(t)}{u_p(t)}=tan(\theta_p(t))\leq\frac{1}{\rho_p}$ \cite{khalil2002nonlinear}. Given that $\tilde{y}_c(t)=\frac{1}{m_{22}}y_{q_c}(t)-\frac{m_{21}}{m_{22}}y_p(t_{p_k})$ takes constant piece-wise values over the time-interval $[t_{p_k},t_{p_{k+1}})$: $||\dot{\tilde{y}}_c(t)||_2=0$ and we can denote $sup_{t \in [t_{p_k},t_{p_{k+1}})}||\tilde{y}_c(t)||_2\leq C_2$. As a results we have $||u_p(t)||_2=||w_1(t)-\tilde{y}_c(t)||_2\leq ||w_1(t)||_2+||\tilde{y}_c(t)||_2 \leq C_1+C_2$, and  we have $||\dot{u}_p(t)||_2=||\dot{w}_1(t)||_2\leq C_0$. Hence:
			\begin{align}
			&\frac{d}{dt}||e_p(t)||_2\leq||\dot{u}_p(t)||_2||tan(\theta_p(t))||_2+\Gamma(2\theta_{k,p})||u_p(t)||_2||tan^2(\theta_p(t))||_2+\Gamma(2\theta_{k,p})||u_p(t)||_2\\
			&~~~~~~~~~~~~~~\leq \frac{C_0}{\rho_p}+\Gamma(2\theta_{k,p})(\frac{1}{\rho_p^2}+1)(C_1+C_2)
			\end{align}
			
			We can find a positive lower-bound for inter-event time intervals between each two triggered instances by integrating from both sides over the time interval $[t_{p_k},t_{p_{k+1}})$ and with the initial condition $e_p(t_{p_k})=0$ given that after each triggering instance the error is reset to zero. We will have:
			
			\begin{align}
			&t_{p_{k+1}}^- - t_{p_k} \geq \frac{||e_p(t_{p_{k+1}}^-)||_2}{\frac{C_0}{\rho_p}+\Gamma(2\theta_{k,p})(\frac{1}{\rho_p^2}+1)(C_1+C_2)}
			\end{align}
			
			The next triggering-condition is met at $t_{p_{k+1}}$ meaning $||e_p(t_{p_{k+1}})||_2^2>\delta_p||y_p(t_{p_{k+1}})||_2^2$, this gives us $||e_p(t_{p_{k+1}}^-)||_2=\sqrt{\delta_p}||y_p(t_{p_{k+1}}^-)||_2$:
			\begin{align}
			&t_{p_{k+1}}^- - t_{p_k} \geq \frac{\sqrt{\delta_p}||y_p(t_{p_{k+1}}^-))||_2}{\frac{C_0}{\rho_p}+\Gamma(2\theta_{k,p})(\frac{1}{\rho_p^2}+1)(C_1+C_2)}
			\end{align}
		
			where 	
			\begin{align}
			&\Gamma(2\theta_{k,p})=||cos^{-1}(\frac{\nu_p+\rho_p}{\sqrt{(1-4\rho_p\nu_p)+(\nu_p+\rho_p)^2}})||_2
			\end{align}
			
			which proves the proposition. 
		\end{proof}
		\begin{remark}
			As it can be seen in (65), the only time $t_{p_{k+1}}^- - t_{p_k}=0$, is when $||y_p(t_{p_{k+1}}^-))||_2=0$. This can be interpreted as the case that the system is at rest or has converged to zero. In these cases, the triggering condition given in (\ref{eq:trigp}) is never met, and no new information is required to be sent. In other cases, such as when the output is tracking a specific signal, we have: $||y_p(t_{p_{k+1}}^-))||_2>0$ and $t_{p_{k+1}}^- - t_{p_k}>0$. However, the designed triggering condition combined with the plant output's rate of change affect the length of inter-event time intervals. 
		\end{remark}
		
		\begin{remark}
			On a similar note, one can see that a larger triggering threshold $\delta_p$ leads to larger inter-event intervals. Same direct relationship holds for output passivity indices, namely, a larger $\rho_p$ also increases the length of inter-event intervals. Moreover, the relationship between passivity indices and inter-event time intervals is explicitly shown in (65). 
		\end{remark}
		
		\begin{remark}
			Once the system is getting closer to convergence, the inter-event time intervals will get shorter in length: the system will need more control actions to guarantee convergence. However, in many cases this process will be fast and short: a faster information exchange means more precise and efficient control actions that will result in a faster convergence.
		\end{remark}
		
		\begin{remark}
			Another important aspect of the formula given in (65) is that $||\dot{w}_1(t)||_2<C_0$ and $sup_{t \in [t_{p_k},t_{p_{k+1}})}||w_1(t)||_2\leq C_1$ have reciprocal relationships with inter-event time intervals, meaning in cases with $w_1$ as the external disturbance to the system, the stronger the disturbance $w_1$ is, the shorter inter-event time intervals are and the triggering condition is met more frequently to compensate for the strong external disturbance. In other words, the relation given in (65) characterized the robustness of the event-triggering condition on the plant's side against external disturbance.
		\end{remark}	
		
		\begin{remark}
			It is also important to note that the results given in proposition 1 are independent of network time-delays, and signal quantizations. However, it is also important to mention that the results in proposition 1 are obtained based on the assumption that the finite-gain stability for the networked control system in figure \ref{fig:ETNCSDPQ} was achieved based on theorem 2 and theorem 3. 
		\end{remark}	

		\subsubsection{Inter-Event Time-Interval Analysis of the Event-Triggering condition on the Controller's Output}

		\begin{proposition}	
				Consider the event-triggered networked control system given in figure \ref{fig:ETNCSDPQ} with time-delays $T_1(t)$ and $T_2(t)$, and quantization blocks $Q_p$ and $Q_c$. Let the plant and controller be Input Feed-forward Output Feedback passive with respective passivity indices $\rho_p$, $\nu_p$, and $\rho_c$, $\nu_c$. Consider also that the transformation matrix $M$ is implemented according to theorem 3.	Let the triggering instances on the plant's, and controller's sides be explicitly determined by the following relations:
				\begin{align*}
				&	||e_p(t)||_2^2>\delta_p||y_p(t)||_2^2~~ where~~ \delta_p  \in (0,1]\\
				&	||e_c(t)||_2^2>\delta_c||y_c(t)||_2^2~~ where~~ \delta_c  \in (0,1]
				\end{align*} 
			For the purpose of analyzing the proposed networked control design's robustness against external disturbance, we assume that an external disturbance $w_2 \neq 0$ which meets the following conditions:\\
			1) $||\dot{w}_2(t)||_2\leq C^\prime_0$ for $t \in [t_{c_k},t_{c_{k+1}})$\\
			2) $sup_{t \in [t_{c_k},t_{c_{k+1}})}||w_2(t)||_2\leq C^\prime_1$\\
			is fed into the networked interconnection on the controller's side. Then for any initial condition  $x_p(0)$ for the plant, the inter-event time interval $\{t_{c_{k+1}} - t_{c_k}\}$ is lower bounded by:
			\begin{align}
			&t_{c_{k+1}}^- - t_{c_k} \geq \frac{\sqrt{\delta_c}||y_c(t_{c_{k+1}}^-))||_2}{\frac{C^\prime_0}{\rho_c}+\Gamma^\prime(2\theta_{k,c})(\frac{1}{\rho_c^2}+1)(C^\prime_1+C^\prime_2)}
			\end{align}
			where 	
			\begin{align}
			&\Gamma^\prime(2\theta_{k,c})=||cos^{-1}(\frac{\nu_c+\rho_c}{\sqrt{(1-4\rho_c\nu_c)+(\nu_c+\rho_c)^2}})||_2
			\end{align}
		\end{proposition}
		\begin{proof}
given in Appendix A.
		\end{proof}
		\begin{remark}
			One can see that a larger triggering threshold $\delta_c$ leads to larger inter-event time intervals. Same direct relationship holds for the output passivity index, namely, a larger $\rho_c$ also increases the length of the inter-event time intervals. Moreover, the relationship between passivity indices and inter-event time intervals are explicitly shown in (67). 
		\end{remark}
		
		\begin{remark}
			Contrary to other proofs, in proposition 2, we aimed to solve the problem by assuming that there is an external disturbance on the controller's side, namely $w_2 \neq 0$. If we assume $w_2 = 0$ then we will have 
			\begin{align}
		&	t_{c_{k+1}}^- - t_{c_k} \geq \frac{\sqrt{\delta_c}||y_c(t_{c_{k+1}}^-))||_2}{\Gamma^\prime(2\theta_{k,c})(\frac{1}{\rho_c^2}+1)C^\prime_2}
			\end{align}
			where 	
			\begin{align}
		&\Gamma(2\theta_{k,c})=||cos^{-1}(\frac{\nu_c+\rho_c}{\sqrt{(1-4\rho_c\nu_c)+(\nu_c+\rho_c)^2}})||_2
			\end{align}
		\end{remark}
		
		\begin{remark}
			Another important aspect of the formula given in (67) is that $||\dot{w}_2(t)||_2<C^\prime_0$ and $sup_{t \in [t_{c_k},t_{c_{k+1}})}||w_2(t)||_2\leq C^\prime_1$ have reciprocal relationships with inter-event time intervals, meaning in cases with $w_2$ as the external disturbance to the system, the stronger the disturbance $w_2$ is, the shorter inter-event time intervals are and the triggering condition is met more frequently to compensate for the strong external disturbance. In other words, the relation given in (67) characterized the robustness of the event-triggering condition on the controller's side against external disturbances.
		\end{remark}	
		
		\begin{remark}
			Similar to proposition 1, the proof for proposition 2 is independent of network time-delays, and signal quantizations. However, the results in proposition 2 are obtained based on the assumption that the networked control system (given in figure \ref{fig:ETNCSDPQ}) is designed based on the criteria given in theorem 2 and theorem 3. 
		\end{remark}	
		\subsection{Robustness of the Proposed Design in the Presence of Data Losses and Packet Dropouts}
		\subsubsection{Maximum Number of Allowed Consecutive Lost Packets - Analysis for the Communication Link from The Plant to Controller}
			In this section, we analyze the robustness of the interconnection presented in figure \ref{fig:ETNCSDPQ} over unreliable band-limited communication networks. We define robustness as the design's ability to perform as expected while tolerating a certain number of consecutive packet dropouts (lost packets) over the communication network. First, we show the proof for the communication link from the plant to controller, and next we show the same analysis for the communication link from the controller to plant. 
			
			Due to the continuous nature of our design, we should clarify what we mean by packet dropouts: on each side, whenever the triggering condition is met, the plant or controller attempts to send new information over the communication link, a packet is dropped if this attempt is not successful. This can be due to lossy communication links, or unsuccessful attempts at the sending or receiving side of the process, but as long as the intention to update old information is not fulfilled, we consider the packet as a lost one. This in return means that most likely, after the respective delays, the triggering condition is met again, and another packet is sent over the communication link. As a result, in the following proofs we assume that the broadcast release times can take values $\{n_i\}_{i=0}^{i=\infty}$ for theorem 3, and $\{n_i^\prime\}_{i=0}^{i=\infty}$ for theorem 4. And the two successful release times before and after each group of packet dropouts are $t_{p_k}$ and $t_{p_{k+1}}$ for theorem 3, and $t_{c_k}$ and $t_{c_{k+1}}$ for theorem 4 where $k$ is a positive integer. With that in mind, now we can proceed to the proofs:
			\subsection{Package dropouts between the plant and controller}
			\begin{theorem}
				Consider the event-triggered networked control system given in figure \ref{fig:ETNCSDPQ} with time-delays $T_1(t)$ and $T_2(t)$, and quantization blocks $Q_p$ and $Q_c$. Let the plant and controller be Input Feed-forward Output Feedback passive with respective passivity indices $\rho_p$, $\nu_p$, and $\rho_c$, $\nu_c$. Consider that the transformation matrix $M$ is implemented according to theorem 3.	Let the triggering instances on the plant's, and controller's sides be explicitly determined by the following relations:
				\begin{align*}
					||e_p(t)||_2^2>\delta_p||y_p(t)||_2^2~~ where~~ \delta_p  \in (0,1]\\
					||e_c(t)||_2^2>\delta_c||y_c(t)||_2^2~~ where~~ \delta_c  \in (0,1]
				\end{align*} 
				Let the number of successive packet losses $d_p \in Z$ between two successful information exchanges $t_{p_k}$ and $t_{p_{k+1}}$ $($for any $t \in [t_{p_k},t_{p_{k+1}}))$ from the plant to controller satisfy
				\[D_p(\tilde{\nu}_c) =
				\begin{cases}
				d_p\leq\lfloor\log_{(1+\sqrt{\delta_p})}^{(\sqrt{\frac{2(\rho_p-\frac{1}{4\gamma})}{\alpha}}+1)}-1\rfloor & \quad \text{if } \tilde{\nu}_c \geq 0\\
				d_p\leq\lfloor\log_{(1+\sqrt{\delta_p})}^{(\sqrt{\frac{2(\rho_p+2\tilde{\nu}_c-\frac{1}{4\gamma})}{\alpha-4\tilde{\nu}_c}}+1)}-1\rfloor  & \quad \text{if } \tilde{\nu}_c < 0\\
				\end{cases}
				\]
				where $\alpha,\gamma>0$, then the networked control system is robust against $d_p$ consecutive packet dropouts and maintains finite-gain $L_2$-stability from the input $w_1(t)$ to the output $y_p(t)$.
			\end{theorem}
			\begin{proof}
				We assume $\delta_p$ was chosen to meet the finite-gain $L_2$-stability requirements given in theorem 2: 			   		$\beta(\tilde{\nu}_c)>\frac{1}{4\gamma}$ and $\tilde{\rho}_c+\nu_p-|\nu_p|-\frac{1}{2\alpha}>0$ where:

				\[\beta(\tilde{\nu}_c) =
				\begin{cases}
				\rho_p-\frac{\delta_p\alpha}{2}      & \quad \text{if } \tilde{\nu}_c \geq 0\\
				\rho_p+2\tilde{\nu}_c-\delta_p(\frac{\alpha}{2}-2\tilde{\nu}_c)  & \quad \text{if } \tilde{\nu}_c < 0\\
				\end{cases}
				\]
				
			    and $\alpha,\gamma>0$, we have:
				\begin{align}
				&||e_p(t)||_2\leq\sqrt{\delta}_p||y_p(t)||_2~~ where~~ \delta_p  \in (0,1]~~~ for~ any~~ t \in [t_{p_k},t_{p_{k+1}})\\
				&||e_c(t)||_2\leq\sqrt{\delta_c}||y_c(t)||_2~~ where~~ \delta_c  \in (0,1]~~~ for~ any~~ t \in [t_{c_k},t_{c_{k+1}})
				\end{align}
				
				We assume $d_p$ packets are lost in the communication network from the plant to controller for the interval $[t_{p_k},t_{p_{k+1}})$, where the communication attempts at instances $t_{p_k}$ and $t_{p_{k+1}}$ are successful. We have: $t_{p_k}=n_0<n_1<...<n_{d_p}<n_{d_p+1}=t_{p_{k+1}}$. Given that we assume all packets between $t_{p_k}$ and $t_{p_{k+1}}$ are lost and there is no successful information exchange between the plant and the controller for this interval, we can assume that the input to the controller is never updated and consequently no new information (control action) is sent back to the plant, hence the holder's output $u_r(t)$ does not change and for any $t \in [n_{d_p},n_{d_p+1}]$ we have:  	
				\begin{align}
				&||e_p(t)||_2=||y_p(t)-y_p(t_{p_k})||_2=||y_p(t)-y_p(n_0)||_2\leq\sum_{i=0}^{d_p-1}||y_p(n_{i+1})-y_p(n_i)||_2+||y_p(t)-y_p(n_{d_p})||_2
				\end{align}
			
				Assuming that the triggering condition holds, and that we have $d_p$ number of packet dropouts in the time interval $[t_{p_k},t_{p_{k+1}})$, for each of the $d_p$ terms in above equation, we have $||y_p(n_i)-y_p(n_{i+1})||_2 \leq \sqrt{\delta_p}||y_p(n_{i+1})||_2 $ for $i=0,...,d_p$. As a result, the following is true for any $t \in [n_{d_p},n_{d_p+1}]$:
				\begin{align}
				&||e_p(t)||_2\leq\sum_{i=0}^{d_p}\sqrt{\delta_p}||y_p(n_{i+1})||_2
				\end{align}
			
				Additionally we have $||e_p(t_{p_{k+1}})||_2=||y_p(n_{d_p})-y_p(n_{d_p+1})||_2\leq\sqrt{\delta_p}||y_p(n_{d_p+1})||_2$, which gives us:
				\begin{align}
				&||y_p(n_{d_p})||_2\leq(1+\sqrt{\delta_p})||y_p(n_{d_p+1})||_2
				\end{align}
				
				and this gives us:
				\begin{align}
				&||y_p(n_1)||_2\leq(1+\sqrt{\delta_p})^{d_p}||y_p(t_{p_{k+1}})||_2\leq(1+\sqrt{\delta_p})^{d_p}||y_p(n_{d_p+1})||_2
				\end{align}
				
				For $d_p$ packet dropouts, we have:
				\begin{align}
				&||e_p(t_{p_{k+1}})||_2\leq\sum_{i=0}^{d_p}
				\sqrt{\delta_p}(1+\sqrt{\delta_p})^i ||y_p(t_{p_{k+1}})||_2\leq\sum_{i=0}^{d_p}
				\sqrt{\delta_p}(1+\sqrt{\delta_p})^i ||y_p(n_{d_p+1})||_2
				\end{align}
				
				(77) is a geometric series with the ratio of $1+\sqrt{\delta_p}$, so we have the following for all $t \in [t_{p_k},t_{p_{k+1}})$:
				\begin{align}
				&||e_p(t_{p_{k+1}})||_2\leq\frac{\sqrt{\delta_p}[1-(1+\sqrt{\delta_p})^{d_p+1}]}{1-(1+\sqrt{\delta_p})}||y_p(t_{p_{k+1}})||_2\leq
				((1+\sqrt{\delta_p})^{d_p+1}-1) ||y_p(t_{p_{k+1}})||_2
				\end{align}
				
				If we look at the conditions given in theorem 2 for finite-gain $L_2$-stability, we have $\beta(\tilde{\nu}_c)>\frac{1}{4\gamma}$ and $\tilde{\rho}_c+\nu_p-|\nu_p|-\frac{1}{2\alpha}>0$ where, 	
				\[\beta(\tilde{\nu}_c) =
				\begin{cases}
				\rho_p-\frac{\delta_p\alpha}{2}      & \quad \text{if } \tilde{\nu}_c \geq 0\\
				\rho_p+2\tilde{\nu}_c-\delta_p(\frac{\alpha}{2}-2\tilde{\nu}_c)  & \quad \text{if } \tilde{\nu}_c < 0\\
				\end{cases}
				\]
				
				and $\alpha,\gamma>0$. For $\tilde{\nu}_c\geq0$, the following should stand for stability conditions to hold:
				
				\begin{align}
				&
				((1+\sqrt{\delta_p})^{d_p+1}-1) <\sqrt{ \frac{2(\rho_p-\frac{1}{4\gamma})}{\alpha}}
				\end{align}

				Hence, we should have: $d_p\leq\lfloor\log_{(1+\sqrt{\delta_p})}^{(\sqrt{ \frac{2(\rho_p-\frac{1}{4\gamma})}{\alpha}}+1)}-1\rfloor$ for $t \in [t_{p_k},t_{p_{k+1}})$ and similarly for the case that $\tilde{\nu}_c<0$ we have: 
				
				\begin{align}
				&
				((1+\sqrt{\delta_p})^{d_p+1}-1) <\sqrt{ \frac{2(\rho_p+2\tilde{\nu}_c-\frac{1}{4\gamma})}{\alpha-4\tilde{\nu}_c}}
				\end{align}
				
				 We should have the allowable number of packet dropouts: $d_p\leq\lfloor\log_{(1+\sqrt{\delta_p})}^{(\sqrt{\frac{2(\rho_p+2\tilde{\nu}_c-\frac{1}{4\gamma})}{\alpha-4\tilde{\nu}_c}}+1)}-1\rfloor$ for all $t \in [t_{p_k},t_{p_{k+1}})$ for the system to stay finite-gain stable.
			\end{proof}
			\begin{remark}
				From (79, 80), the number of allowable packet dropouts from the plant to controller has a direct relation with output passivity index $\rho_p$. This means that a networked control system with a larger plant's output passivity index will stay stable for longer time-intervals of information loss.
			\end{remark}
			\begin{remark}
				The number of allowed consecutive lost packets in a networked control system can be seen as a measure of its robustness. Namely, if a system can stay stable under a certain time-interval of lost communication, then the specific networked control design is robust against information loss.
			\end{remark}
			\begin{remark}
				Another point to make here is that a negative $\tilde{\nu}_c$ can greatly harm the robustness of the networked control system, however this should not be an issue in our design, given that we can design the $M$-Transformation matrix so that we always have a positive $\tilde{\nu}_c$.
			\end{remark}
				\begin{remark}
				Unlike most approaches in the literature, we did not assume a pre-known maximum or minimum number of packet dropouts in the communication link, our result sought to show a design-based trade-off between parameters, and robustness of the design against information loss.
				\end{remark}
		\subsubsection{Maximum Number of Allowed Consecutive Lost Packets - Analysis for the Communication Link from The Controller to Plant}

			\begin{theorem}
			Consider the event-triggered networked control system given in figure \ref{fig:ETNCSDPQ} with time-delays $T_1(t)$ and $T_2(t)$, and quantization blocks $Q_p$ and $Q_c$. Let the plant and controller be Input Feed-forward Output Feedback passive with respective passivity indices $\rho_p$, $\nu_p$, and $\rho_c$, $\nu_c$. Consider that the transformation matrix $M$ is implemented according to theorem 3.	Let the triggering instances on the plant's, and controller's sides be explicitly determined by the following relations:
			\begin{align*}
				||e_p(t)||_2^2>\delta_p||y_p(t)||_2^2~~ where~~ \delta_p  \in (0,1]\\
				||e_c(t)||_2^2>\delta_c||y_c(t)||_2^2~~ where~~ \delta_c  \in (0,1]
			\end{align*} 
				Let the number of successive packet losses $d_c \in Z$ between two successful information exchanges $t_{c_k}$ and $t_{c_{k+1}}$$($for any $t \in [t_{c_k},t_{c_{k+1}}))$ from the controller to plant satisfy
				\begin{align*}
				d_c\leq\lfloor\log_{(1+\sqrt{\delta_c})}^{\sqrt{\frac{m_{22}^2\rho_c}{(\frac{1}{2\alpha}+|\nu_p|-\nu_p)2b_c^2(1+d_2)}}+1)}-1\rfloor
				\end{align*}
				where $\alpha,\gamma>0$, then the networked control system is robust against $d_c$ consecutive packet dropouts and maintains finite-gain $L_2$-stability from the input $w_1(t)$ to output $y_p(t)$.
			\end{theorem}
			\begin{proof}
				We assume $\delta_c$ was chosen to meet the finite-gain $L_2$-stability requirements given in theorem 2 and design requirements in theorem 3, as a result we have the following upper bound for $(1+\sqrt{\delta_c})$: 
				
				\begin{align*}
				&m_{22}^2 > \frac{(\frac{1}{2\alpha}+|\nu_p|-\nu_p)(2b_c^2(1+\sqrt{\delta_c})^2(1+d_2))}{\rho_c}\\
				& \sqrt{\frac{m_{22}^2\rho_c}{(\frac{1}{2\alpha}+|\nu_p|-\nu_p)2b_c^2(1+d_2)}}>(1+\sqrt{\delta_c})
				\end{align*}				   	
				
				and $\alpha,\gamma>0$. We also have:
				\begin{align}
				&||e_p(t)||_2\leq\sqrt{\delta}_p||y_p(t)||_2~~ where~~ \delta_p  \in (0,1]~~~ for~ any~~ t \in [t_{p_k},t_{p_{k+1}})\\
				&||e_c(t)||_2\leq\sqrt{\delta_c}||y_c(t)||_2~~ where~~ \delta_c  \in (0,1]~~~ for~ any~~ t \in [t_{c_k},t_{c_{k+1}})
				\end{align}
				
				We consider $d_c$ as the number of packets that are lost in the communication network from the controller to plant for the interval $[t_{c_k},t_{c_{k+1}})$, so we have: $t_{c_k}=n^\prime_0<n^\prime_1<...<n^\prime_{d_c}<n^\prime_{d_c+1}=t_{c_{k+1}}$. For any $t \in [n^\prime_{d_c},n^\prime_{d_c+1}]$ we have:  	
				\begin{align}
				&||e_c(t)||_2=||y_c(t)-y_c(t_{c_k})||_2=||y_c(t)-y_c(n^\prime_0)||_2\leq\sum_{i=0}^{d_c-1}||y_c(n^\prime_{i+1})-y_c(n^\prime_i)||_2+||y_c(t)-y_c(n^\prime_{d_c})||_2
				\end{align}
				
				Assuming that the triggering condition holds, and that we have $d_c$ number of packet dropouts in the time interval $[t_{c_k},t_{c_{k+1}})$, for each of the $d_c$ terms in above equation, we have $||y_c(n^\prime_i)-y_c(n^\prime_{i+1})||_2 \leq \sqrt{\delta_c}||y_c(n^\prime_{i+1})||_2$ for $i=0,...,d_c$. As a result, the following is true for any $t \in [n^\prime_{d_c},n^\prime_{d_c+1}]$:
			
				\begin{align}
				&||e_c(t)||_2\leq\sum_{i=0}^{d_c}\sqrt{\delta_c}||y_c(n^\prime_{i+1})||_2
				\end{align}
				
				Additionally, we have $||e_c(t_{c_{k+1}})||_2=||y_c(t_{c_{k+1}})-y_c(n^\prime_{d_c})||_2\leq||y_c(n^\prime_{d_c})-y_c(n^\prime_{d_c+1})||_2\leq\sqrt{\delta_c}||y_c(n^\prime_{d_c+1})||_2$, which gives us:
				\begin{align}
				&||y_c(n^\prime_{d_{c}})||_2\leq(1+\sqrt{\delta_c})||y_c(n^\prime_{d_c+1})||_2
				\end{align}
				
				and:
				\begin{align}
				&||y_c((n^\prime_1)||_2\leq(1+\sqrt{\delta_c})^{d_c}||y_c(t_{c_{k+1}})||_2\leq(1+\sqrt{\delta_c})^{d_c}||y_c(n^\prime_{d_{c}+1})||_2
				\end{align}
				
				For $d_c$ packet dropouts we have:
				\begin{align}
				&||e_c(t_{c_{k+1}})||_2\leq\sum_{i=0}^{d_c}
				\sqrt{\delta_c}(1+\sqrt{\delta_c})^i ||y_c(t_{c_{k+1}})||_2\leq\sum_{i=0}^{d_c}
				\sqrt{\delta_c}(1+\sqrt{\delta_c})^i ||y_c(n^\prime_{d_{c}+1})||_2
				\end{align}
				
				(87) is a geometric series with the ratio of $1+\sqrt{\delta_c}$, so we have the following for all $t \in [t_{c_k},t_{c_{k+1}})$:
				\begin{align}
				&||e_c(t_{c_{k+1}})||_2\leq\frac{\sqrt{\delta_c}[1-(1+\sqrt{\delta_c})^{d_c+1}]}{1-(1+\sqrt{\delta_c})}||y_c(t_{c_{k+1}})||_2\leq
				((1+\sqrt{\delta_c})^{d_c+1}-1) ||y_c(t_{c_{k+1}})||_2
				\end{align}
				
				If we look at the conditions given in theorem 2, and theorem 3, we should have 
				
				\begin{align}
				&
				((1+\sqrt{\delta_c})^{d_c+1}-1) <\sqrt{\frac{m_{22}^2\rho_c}{(\frac{1}{2\alpha}+|\nu_p|-\nu_p)2b_c^2(1+d_2)}}
				\end{align}
				
				For (89) to hold, we should have: $d_c\leq\lfloor\log_{(1+\sqrt{\delta_c})}^{\sqrt{\frac{m_{22}^2\rho_c}{(\frac{1}{2\alpha}+|\nu_p|-\nu_p)2b_c^2(1+d_2)}}+1)}-1\rfloor$ for any $t \in [t_{c_k},t_{c_{k+1}})$ which proves the theorem.
				
			\end{proof}
			\begin{remark}
				As you can see from the results, the number of allowed consecutive packet dropouts for the communication link from the controller to plant has a direct relation with output passivity index $\rho_c$, namely, a networked control system with a controller with a larger output passivity index stays stable under longer time-intervals of information loss.
			\end{remark}
			\begin{remark}
				The number of allowed lost packets in a controlled networked system can be seen as a measure of its robustness. Namely, if a system can stay stable under a certain time-interval of lost communication, then the specific networked control design is robust against information loss.
			\end{remark}
			\begin{remark}
				(89) characterizes the relationship between $d_c$ and $m_{22}^2$. This means that the designer of the $M$-Transformation matrix, by selecting $m_{22}^2$, has some control over the interconnection's robustness against packet losses in the communication link from the controller to plant. 
		\end{remark}
		\begin{remark}
				Comparing the relations given in theorem 4 and theorem 5, we find out that our proposed networked control design is more tolerate toward packet losses in the communication link from the controller to plant, compared to packet losses in the communication link from the plant to controller. This is in line with a similar observation made in theorem 3, where the triggering condition in the communication link from the controller to plant is usually less conservative resulting in a lower communication rate, compared to the triggering condition and communication rate in the communication link from the plant to controller. 
			\end{remark}
				\begin{remark}
					Similar to the previous proof, we did not assume a pre-known maximum or minimum number of packet dropouts in the communication link, our result shows a design-based trade-off between parameters, and robustness of the design against information loss. For example, for more unreliable communication links, it would be more prudent to select a larger value for $m_{22}$, the trade-off in this case is that this will result in a smaller input passivity index $\tilde{\nu}_c$, however the design will be more robust against packet dropouts.
				\end{remark}
			\clearpage
\section{Simulation Results}
\label{sec:simulation}
We consider an event-triggered networked control interconnection of two Input Feed-forward Output Feedback passive (IFOFP) systems over a band-limited networked interconnection that follows the structure given in figure \ref{fig:ETNCSDPQ}. We assume the following network induced delays: an initial time-delay of $T_1(0)=0.5 s$ with an increasing rate of change of $d_1=0.3$ in the communication link from the plant to controller, and an initial time-delay of $T_2(0)=0.6s$ with an increasing rate of change of $d_2=0.2$ in the communication link from the controller to plant. The quantizers used on both sides are uniform mid-tread quantizers with quantization level of 0.5 ($a_c=a_p=0$, $b_c=b_p=2$). Additionally, we assume that the network interconnections are lossy and not all communication attempts are successful, as a result we assume that there are random packet dropouts between two successful information exchanges.

The nonlinear IFOFP plant has the dynamics:
	\begin{align*} 
	&\dot{x}_{p1}(t)=-3x^3_{p1}(t)+x_{p1}(t)x_{p2}(t)\\
	&\dot{x}_{p2}(t)=-3.6x_{p2}(t)+2u_p(t)\\
	&~y_p(t)=x_{p2}(t),
	\end{align*}
	given the storage function $V(x)=\frac{1}{4}x_{p2}^2$, the plant has passivity levels $\rho_p=1.8$ and $\nu_p=0$. The initial states of the plant are $x_{p1}(0)=10$, $x_{p2}(0)=-14$. 
	
	The  IFOFP linear controller has the dynamics:
	\begin{align*} 
	&\dot{x}_{c1}(t)=-3x_{c}(t)+u_c(t)\\
	&~ y_c(t)=7x_{c}(t)+u_c(t),
	\end{align*}
	with passivity levels $\rho_c=0.27$ and $\nu_c=0.49$ according to its Nyquist and inverse Nyquist plots (the Nyquist and inverse Nyquist plots of a passive linear system with transfer function $G(s)$ lie on the closed right half-side of the complex plane and their intersections with the real axis determine passivity indices of the system \cite{bao2007process}).
	
	Now we will choose our design parameters according to theorem 2 and theorem 3. First, we select $\delta_c=0.15$ for the event-triggering condition on the controller's side. Next, we pick $\gamma=\frac{1000}{4}$ and $\alpha=1$, assuming that we will design the entries for the $M$-transformation matrix to have $\tilde{\nu}_c>0$, and given theorem 2, by selecting $\delta_p=0.4$ we will have $\beta(\tilde{\nu}_c)>\frac{1}{4\gamma}$. Our triggering conditions become:
	\begin{align*}
	&||e_p(t)||_2^2>0.40||y_p(t)||_2^2\\
	&||e_c(t)||_2^2>0.15||y_c(t)||_2^2
	\end{align*} 
	Next, we need to select $m^2_{22}$ such that $\tilde{\rho}_c+\nu_p-|\nu_p|-\frac{1}{2\alpha}>0$ where the relationship between $\tilde{\rho}_c$ and $m^2_{22}$ is given in theorem 3. We choose $m^2_{22}=49.46$, and consequently have $m_{22}=7.03$, this gives us $m_{21}=-4.86$, and finally we select $m_{11}=0.16$ to have the positive $\tilde{\nu}_c=0.03$ and $\tilde{\rho}_c=0.72$. According to theorems 4 and 5, our specific design is robust against random single packet dropouts between each successful information exchange in the communication link from the plant to controller, and the communication link from the controller to the plant: $d_p=1\leq\lfloor\log_{1.63}^{(\sqrt{3.59}+1)}-1\rfloor $ and $d_c=2\leq\lfloor\log_{1.38}^{(\sqrt{2.78}+1)}-1\rfloor $. In our simulation, we assumed a uniform distribution for packet dropouts with the probability of $\frac{1}{2}$ for unsuccessful information transformations. The distribution of successful/unsuccessful attempts over the course of our experiment is depicted in figure \ref{fig:ETDO}: we can see that the communication link from the plant to controller is robust against one packet dropout, and the communication link from the controller to plant is robust against 2 consecutive packet dropouts. Additionally, the inter-event times for the event-triggering condition on the plant's side and controller's side are given in figure \ref{fig:ETDO}.  The external disturbance signal $w_1(t)$ applied to the plant follows a randomly uniform distribution on the interval $[0,2]$. The states response for our design is given in figure \ref{fig:states}: given $y_p(t)=x_{p2}(t)$, we can see that the networked control system is finite-gain $L_2$-stable. Figure \ref{fig:outputs} gives a perspective into the effects of time-varying delays on the outputs from the plant to controller, and the controller to plant: it is important to mention that after delay signals depicted in figure \ref{fig:outputs} are actual signals across the ZOH blocks that are fed into the controller and plant, and the before delay signals show what the expected evolution of the output signals would have been in the absence of the network induced delays. As our last remark, we compare figures \ref{fig:ETDO}, and \ref{fig:outputs} to see the effects of quantization; namely, while from figure \ref{fig:ETDO} we observe that the triggering condition is still met several times toward the end of the experiment, from figure \ref{fig:outputs} we can observe that control actions received by the plant do not change. This is due to the fact that the magnitudes of updated control actions are small and close enough that they lie on the same quantization interval resulting in the same final quantized values.
		\begin{figure}[!t]
			\centering
			\includegraphics[scale = .55]{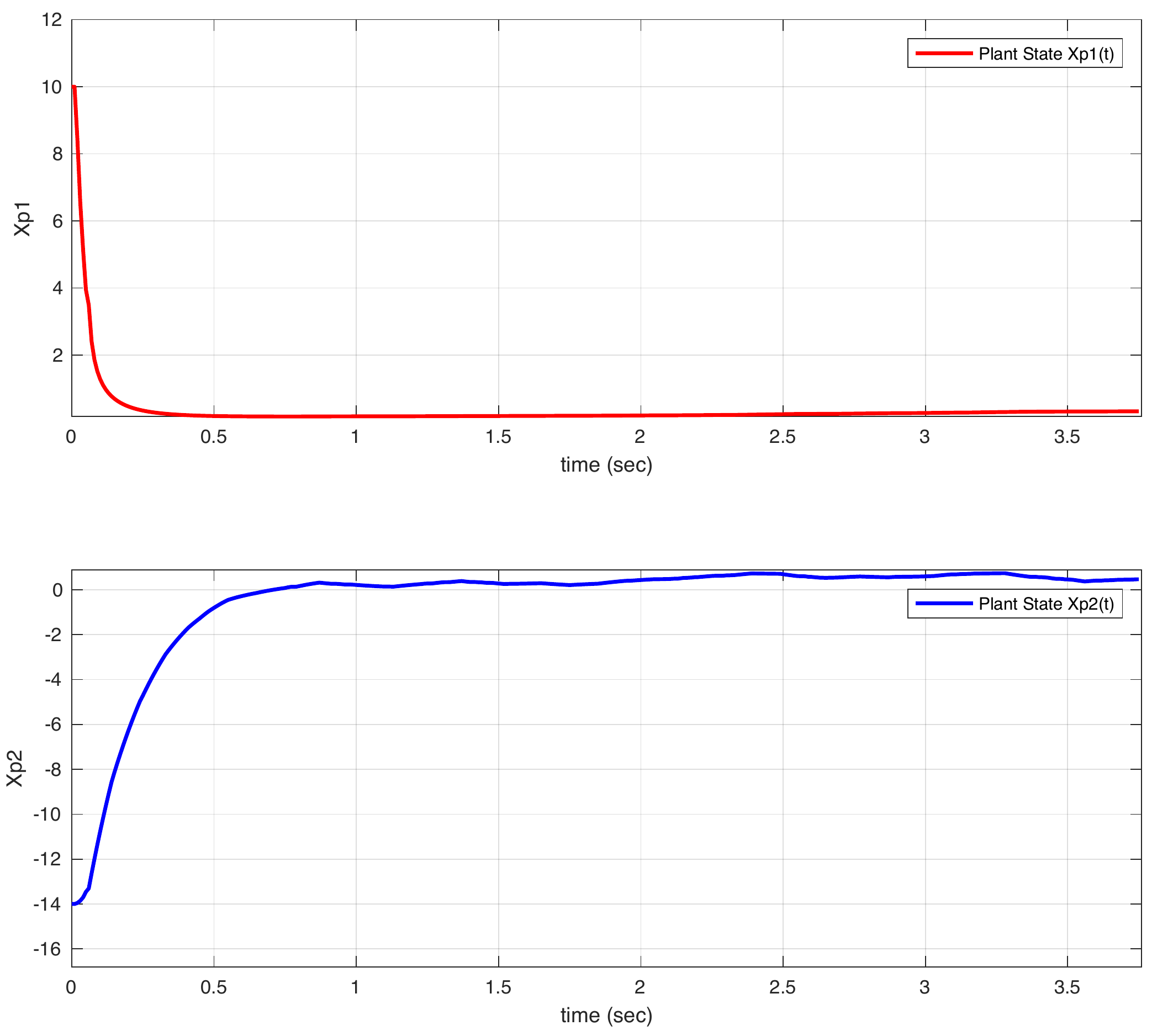}
			\caption{The States Response for the proposed Networked Control Design.}
			\label{fig:states}
		\end{figure}
			\begin{figure}[!t]
				\centering
				\includegraphics[scale = .5]{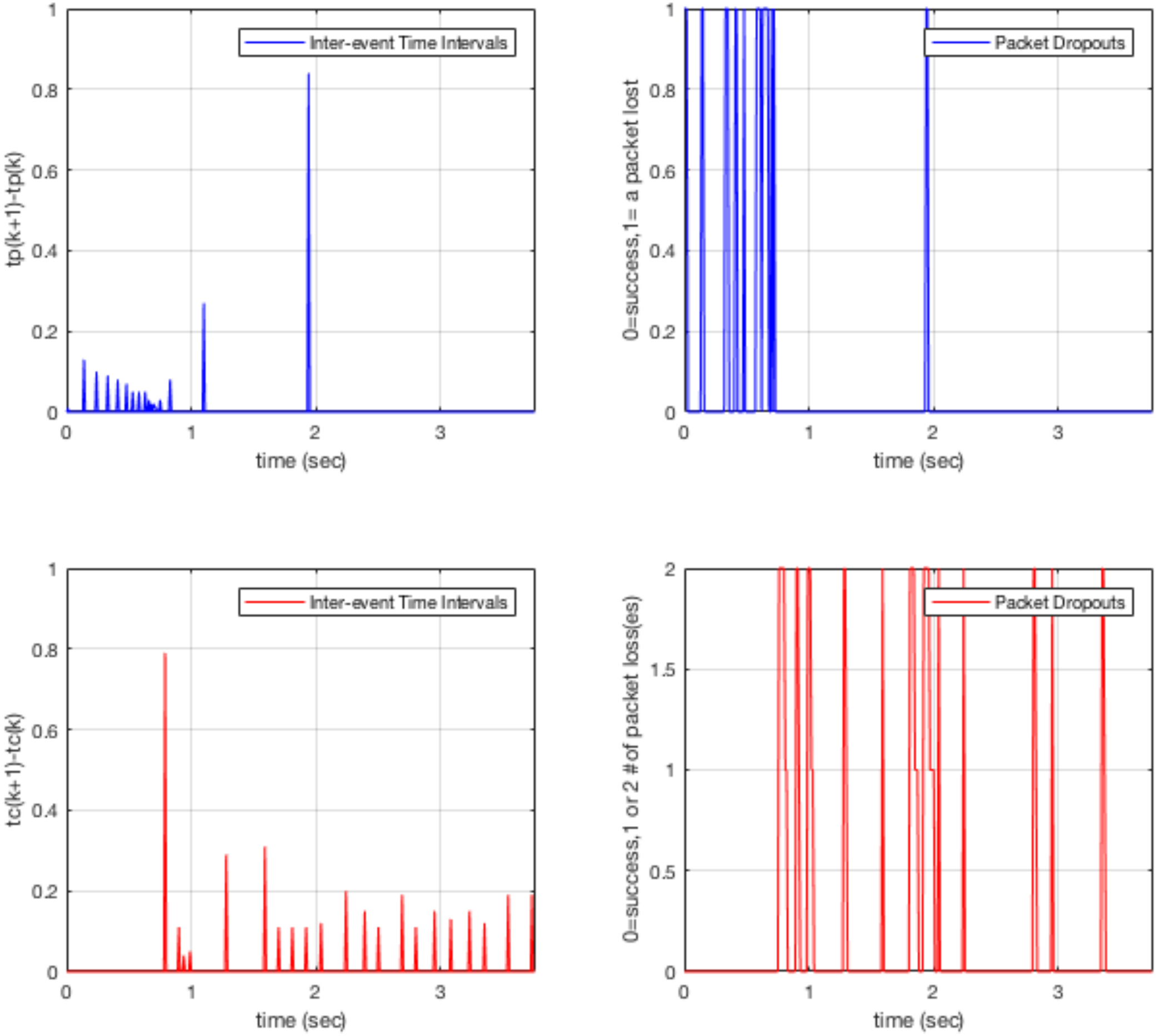}
				\caption{The evolution of Inter-event Time Intervals and Packet Dropouts.}
				\label{fig:ETDO}
			\end{figure}
				\begin{figure}[!t]
					\centering
					\includegraphics[scale = .5]{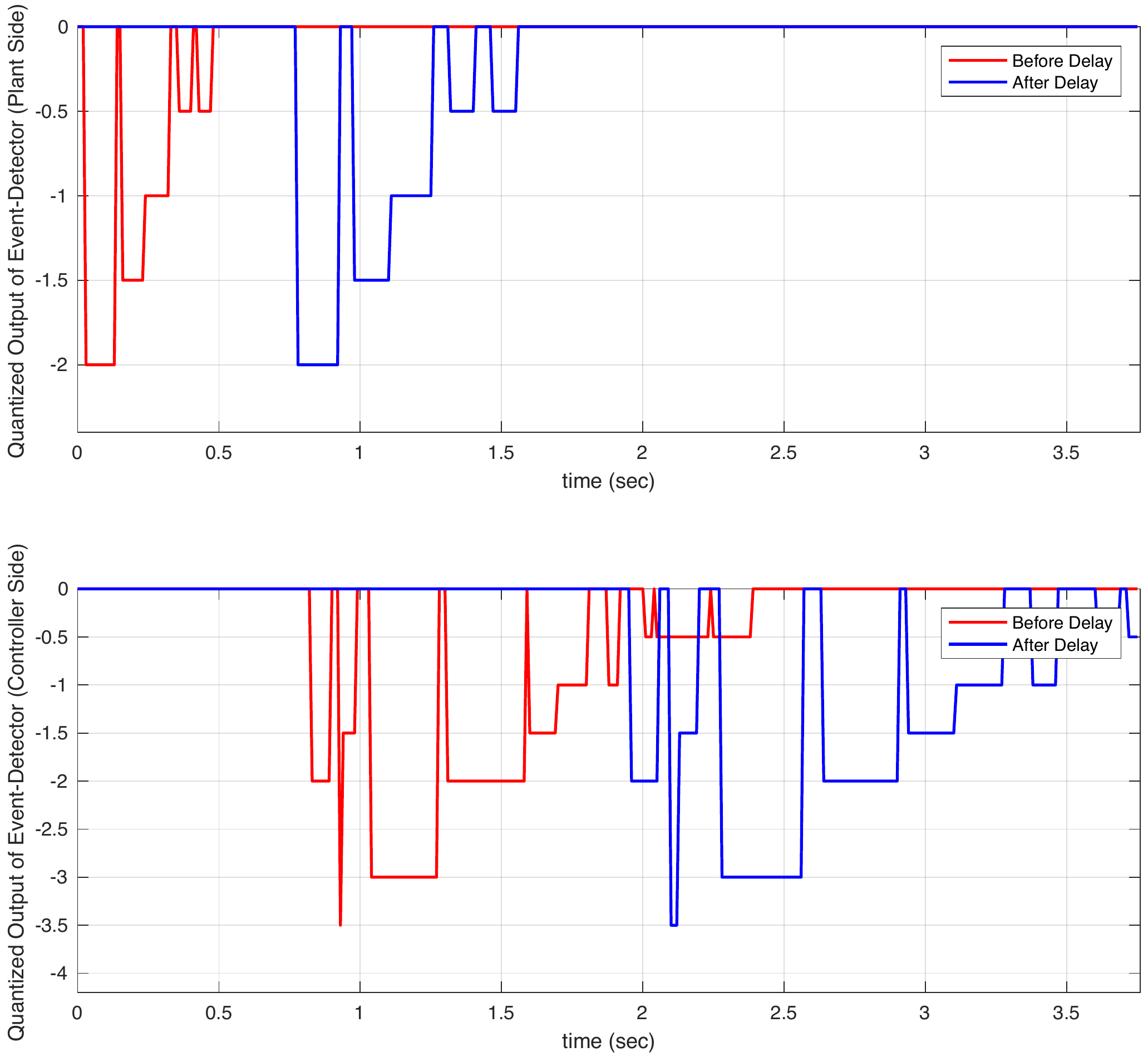}
					\caption{The evolution of Quantized plant's and controller's outputs with time-varying network induced time-delays .}
					\label{fig:outputs}
				\end{figure}
\clearpage			
\section{Conclusion}
\label{sec:conclusion}		
In this report, we intended to propose a design-based framework for finite-gain $L_2$-stability of event-triggered networked control systems. We addressed the robustness issues as well as stability conditions for our design. To best of our knowledge, our proposed design appears to be the most comprehensive approach existing in the networked control literature, in the sense that we consider the effects of signal quantizations, time-varying delays and packet dropouts altogether. Additionally, we proposed a novel conic-based analysis for inter-event time intervals in event-triggered networked control systems consisting of Input Feed-forward Output Feedback subsystems to show the robustness of our design against undesirable "Zeno" behaviors. Our work presented in this report is important because we propose stability conditions for event-triggered designs in networked control systems by considering event-triggering conditions on both sides of the interconnection; this leads to a considerable decreases in the communication load on shared networks. Moreover, we address the effects of signal quantizations and time-varying delays which are unavoidable when information is sent over communication networks. One benefit of our passivity-based approach is its sole reliance on input-output relations of subsystems, this means that we can characterize the stability conditions for a large class of systems with unknown exact dynamics. Contrary to existing input-to-state stability results for NCS in the literature which require systems with only observable states, our stability results can be applied to dynamical systems with unobservable states as well. In addition, our approach is design-based which leaves a lot of leeway for the designer to select different design parameters based on stability, desirable communication rate, performance and robustness requirements. 
\clearpage

\appendix
\section{Proof for Proposition 2} \label{App:AppendixA}

		\begin{proof}
			Our proof is similar to our previous proof for the lower-bound of plant's inter-event time intervals. $G_c$ in figure \ref{fig:ETNCSDPQ} with  passivity indices of $\nu_c$, $\rho_c$ occupy a conic sector in the input-output space which is defined by its center line $\theta_z$ and its apex angle $2\theta_{k,c}$. More importantly, at any instance $t$, the input and output lie within the conic sector:
			\begin{align}
			&\theta_c(t) \in [\theta_z-\theta_{k,c},\theta_z+\theta_{k,c}]
			\end{align}
			
			Additionally, we can define the relationship between the input and output with $\theta_c(t)$ by parameterizing them in polar coordinates:
			
			\begin{align}
			&u_c(t)=r_c(t)cos(\theta_c(t))\\
			&y_c(t)=r_c(t)sin(\theta_c(t)).
			\end{align}
			
			Furthermore, we have:
			\begin{align}
			&cot(2\theta_z)=\nu_c-\rho_c,~~ \theta_z \in [0, \frac{\pi}{2}]\\
			&cos(2\theta_{k,c})=\frac{\nu_c+\rho_c}{\sqrt{(1-4\rho_c\nu_c)+(\nu_c+\rho_c)^2}},~~\theta_{k,c} \in [0,\frac{\pi}{2}).
			\end{align}
			
			Since $e_c(t)=y_c(t)-y_c(t_{c_k})$ for all $t \in [t_{c_k},t_{c_{k+1}})$, we have the following for any $t \in [t_{c_k},t_{c_{k+1}})$:
			\begin{align}
			&\frac{d}{dt}||e_c(t)||_2\leq||\dot{e}_c(t)||_2=||\dot{y}_c(t)||_2=||\dot{r}_c(t)sin(\theta_c(t))||_2+||r_c(t)\dot{\theta}_c(t)cos(\theta_c(t))||_2.
			\end{align}
			
			Given (90), we have $||\dot{\theta}_c(t)||_2\leq||2\theta_{k,c}||_2$, where:
			\begin{align}
			&||2\theta_{k,c}||_2=||cos^{-1}(\frac{\nu_c+\rho_c}{\sqrt{(1-4\rho_c\nu_c)+(\nu_c+\rho_c)^2}})||_2=\Gamma^\prime(2\theta_{k,c})
			\end{align}
			
			Hence we have:
			\begin{align}
			&\frac{d}{dt}||e_c(t)||_2\leq||\dot{e}_c(t)||_2=||\dot{y}_c(t)||_2=||\dot{r}_c(t)sin(\theta_c(t))||_2+||r_c(t)\dot{\theta}_c(t)cos(\theta_c(t))||_2\\
			&~~~~~~~~~~~~~~\leq||\dot{r}_c(t)sin(\theta_c(t))||_2+\Gamma^\prime(2\theta_{k,c})||u_c(t)||_2.
			\end{align}
			
			Furthermore, we have:
			\begin{align}
			&||\dot{r}_c(t)||_2=||\frac{d}{dt}(\frac{u_c(t)}{cos(\theta_c(t))})||_2=||\frac{\dot{u}_c(t)cos(\theta_c(t))+u_c(t)\dot{\theta}_c(t)sin(\theta_c(t))}{cos^2(\theta_c(t))}||_2
			\end{align}	
			
			Inserting the above in (98), we have:
			\begin{align}
			&\frac{d}{dt}||e_c(t)||_2\leq||\dot{u}_c(t)tan(\theta_c(t))+u_c(t)\Gamma^\prime(2\theta_{k,c})tan^2(\theta_c(t))||_2+\Gamma^\prime(2\theta_{k,c})||u_c(t)||_2
			\end{align}
			
			Given $\rho_c>0$, we have $\frac{y_c(t)}{u_c(t)}=tan(\theta_c(t))\leq\frac{1}{\rho_c}$ \cite{khalil2002nonlinear}. Given that $y_{q_p}(t)$ takes constant piece-wise values over the time-interval $[t_{c_k},t_{c_{k+1}})$: $||\dot{y}_{q_p}(t)||_2=0$ and we can denote $sup_{t \in [t_{c_k},t_{c_{k+1}})}||y_{q_p}(t)||_2\leq C^\prime_2$ (it is important to note that the proof does not depend on the network delay, $y_{q_p}(t)$ can encompass the delayed or non-delayed version of the signal depending on the network). As a results we have $||u_c(t)||_2=||w_2(t)+y_{q_p}(t)||_2\leq ||w_2(t)||_2+||y_{q_p}(t)||_2 \leq C^\prime_1+C^\prime_2$, and  we have $||\dot{u}_c(t)||_2=||\dot{w}_2(t)||_2\leq C^\prime_0$. We have:
			\begin{align}
			&\frac{d}{dt}||e_c(t)||_2\leq||\dot{u}_c(t)||_2||tan(\theta_c(t))||_2+\Gamma^\prime(2\theta_{k,c})||u_c(t)||_2||tan^2(\theta_c(t))||_2+\Gamma^\prime(2\theta_{k,c})||u_c(t)||_2\\
			&~~~~~~~~~~~~~~\leq \frac{C^\prime_0}{\rho_c}+\Gamma^\prime(2\theta_{k,c})(\frac{1}{\rho_c^2}+1)(C^\prime_1+C^\prime_2)
			\end{align}
			
			We can find a positive lower-bound for inter-event time intervals between each two triggered instances by integrating from both sides over the time interval $[t_{c_k},t_{c_{k+1}})$ and with the initial condition $e_c(t_{c_k})=0$ given that after each triggering instance the error is reset to zero. We will have:
			
			\begin{align}
			&	t_{c_{k+1}}^- - t_{c_k} \geq \frac{||e_c(t_{c_{k+1}}^-)||_2}{\frac{C^\prime_0}{\rho_c}+\Gamma^\prime(2\theta_{k,c})(\frac{1}{\rho_c^2}+1)(C^\prime_1+C^\prime_2)}
			\end{align}
			
			The next triggering-condition is met at $t_{c_{k+1}}$ meaning $||e_c(t_{c_{k+1}})||_2^2>\delta_c||y_c(t_{c_{k+1}})||_2^2$, this gives us $||e_c(t_{p_{k+1}}^-)||_2=\sqrt{\delta_c}||y_c(t_{c_{k+1}}^-)||_2$:
			\begin{align}
			&t_{c_{k+1}}^- - t_{c_k} \geq \frac{\sqrt{\delta_c}||y_c(t_{c_{k+1}}^-))||_2}{\frac{C^\prime_0}{\rho_c}+\Gamma^\prime(2\theta_{k,c})(\frac{1}{\rho_c^2}+1)(C^\prime_1+C^\prime_2)}
			\end{align}
			
			where 	
			\begin{align}
			&\Gamma(2\theta_{k,c})=||cos^{-1}(\frac{\nu_c+\rho_c}{\sqrt{(1-4\rho_c\nu_c)+(\nu_c+\rho_c)^2}})||_2
			\end{align}
			
			which proves the proposition. 
		\end{proof}

\clearpage
\bibliographystyle{IEEEtran}
\bibliography{bibfile}
\end{document}